\UseRawInputEncoding 
\documentclass[journal]{IEEEtran}
\usepackage{stfloats}
\usepackage{makecell}
\usepackage{graphicx}
\usepackage[cmex10]{amsmath}
\usepackage{cases}
\usepackage[tight,footnotesize]{subfigure}
\usepackage{amsthm} 
\usepackage{cite}
\usepackage{citesort}
\usepackage{amssymb}
\allowdisplaybreaks[4]
\usepackage{algorithm}
\usepackage{algorithmic}
\usepackage{multirow}
\usepackage{amsmath}
\usepackage{xcolor}
\usepackage{CJK}
\usepackage{subeqnarray}
\usepackage{cases}
\usepackage{enumerate}

\usepackage{stfloats}
\usepackage{multirow}

\newtheorem{lemma}{\hskip\parindent\bf{Lemma}}

\ifCLASSINFOpdf
\else
\fi

\begin{document}

\title{MIMO-OFDM Dual-Functional Radar-Communication Systems: Low-PAPR Waveform Design}

\author{Xiaoyan Hu,~\IEEEmembership{Member,~IEEE,} Christos Masouros,~\IEEEmembership{Senior Member,~IEEE,} \\ Fan Liu,~\IEEEmembership{Member,~IEEE,}
Ronald Nissel,~\IEEEmembership{Member,~IEEE,} 

\thanks{X. Hu and C. Masouros are with the Department of Electronic and Electrical Engineering, University College London, London  WC1E 7JE, UK (Email: $\{\rm xiaoyan.hu, c.masouros\}@ucl.ac.uk$).}
\thanks{Fan Liu  is with the Department of Electical and Electronic Engineering, Southern University of Science and Technology, Shenzhen, China ( Email: $\rm liuf6@sustech.edu.cn$).}
\thanks{R. Nissel is with Huawei Technologies, Gothenburg, Sweden (Email: $\rm ronald.nissel@huawei.com$).}
}
\maketitle
\begin{abstract}
In this paper, we explore a dual-functional radar-communication (DFRC) system for achieving integrated sensing and communications (ISAC). The technique of orthogonal frequency division multiplexing (OFDM) is leveraged to overcome the frequency-selective fading of the wideband multiple-input multiple-output (MIMO) systems with one multi-antenna DFRC base station (BS) and multiple single-antenna user equipment (UEs). In order to restrain the high peak-to-average power ratio (PAPR) of OFDM signals, we aim to jointly design low-PAPR DFRC MIMO-OFDM waveforms. This is done by utilizing a weighted objective function on both communication and radar performance metrics  under power and PAPR constraints. The formulated optimization problems can be equivalently transformed into standard semi-definite programming (SDP) and can be effectively solved by semi-definite relaxation (SDR) method, where we prove that globally optimal rank-1 solution can be obtained in general. We further develop a low-complexity method  to solve the problems  with much reduced overheads. 
Moreover, the practical scenario with oversampling on OFDM signals is further considered, which has a significant effect on the resulting PAPR levels.
The  feasibility, effectiveness, and flexibility of the proposed low-PAPR DFRC MIMO-OFDM waveform design methods are demonstrated by a range of simulations on communication sum rate, symbol error rate as well as radar beampattern and detection probability.
\end{abstract}
\begin{IEEEkeywords}
Spectrum sharing, DFRC, MIMO radar, communications,  OFDM, PAPR 
\end{IEEEkeywords}

\IEEEpeerreviewmaketitle

\vspace{-2mm}
\section{Introduction}\label{sec:Introduction}
\subsection{Motivations and Prior Works}\label{sec:PW}
The explosive growth of mobile and Internet-of-things (IoT) devices along with the severe spectrum shortage have driven the demand for an ever-increasing spectrum usage efficiency. To meet the impending need of massive connectivity in 2025 with around 75 billion devices worldwide \cite{2016Electronics_P.Brown_Billion},  sharing/resuing spectrum resources that have already been assigned to other existing applications is regarded as a promising strategy for future communication networks. Among the available applications, radar spectrum is widely regarded as a promising candidate that can facilitate communication and radar spectrum sharing (CRSS) due to the fact that the radar spectrum provides large frequency bands that are suitable for wireless communications, and thus a win-win spectrum sharing policy can be achieved for both sides \cite{2015Proceedings_HGriffithsRadar}. 
As a step further, with sharing one set of hardware equipment and signal processing frameworks for communications and radar,  dual-functional radar-communication (DFRC) design provides a cost-efficient way for achieving integrated sensing and communications (ISAC) \cite{2019SPM_A.Hassanien_Dual,2020SPM_D.Ma_Joint}. In fact, the techniques of DFRC and ISAC are eagerly required by  future intelligent 6G and IoT applications, such as smart cities, smart homes, automatic driving, industry 4.0, etc., which have drawn great attention from both academia and industry recently \cite{2019SPM_A.Hassanien_Dual,2020SPM_D.Ma_Joint,2021TCOMF.Liu_JointRadar,2021ArXiv_A.Zhang_AnOverview}. 

DFRC systems aim at fulfilling wireless communications and radar detections simultaneously  through designing a single transmitted waveform. For \emph{radar-centric \mbox{DFRC} systems} \cite{2019SPM_A.Hassanien_Dual,1963TSET_R.Mealey_AMethod,2007TMTT_G.Saddik_Ultra,2016TSP_A.Hassanien_Dual}, standard radar waveforms are employed and information is introduced by modulating and signaling for communications.
In comparison,  for \emph{communication-centric DFRC systems} \cite{2020SPM_D.Ma_Joint,2011PIEEE_C.Sturm_Waveform,2016RadarConf_D.Gaglione_Fractional,2018TVT_P.Kumari_IEEE802}, the communication waveforms are directly exploited for sensing by extracting the radar information from targets echoes. 
It is known that radar-centric DFRC systems are capable of providing desirable sensing performance but with limited communication rates below 5G requirements, while communication-centric systems can achieve favorable communication performance but with unreliable and difficult-to-tune sensing performance \cite{2021TCOMF.Liu_JointRadar,2021ArXiv_A.Zhang_AnOverview}.
Recently, more attention is focused on \emph{DFRC systems with joint DFRC waveform design} \cite{2018TWC_F.Liu_MU_MIMO,TSP2018_F.Liu_Toward,2020TSP_X.Liu_Joint} to guarantee both sensing and communication performance, which is not limited by any existing radar or communication waveforms and is promising to achieve scalable performance trade-off between the two functionalities.

The technique of orthogonal frequency division multiplexing (OFDM), as a key enabler for 4G and 5G wireless networks, is part of communication standards and has been studied extensively \cite{2004PIEEE_G.Stuber_Broadband,B_D.Tse2005Fundamentals,B2010_Y.Cho_MIMOOFDM}. Actually, OFDM has also recently been exploited for radar sensing  \cite{2011PIEEE_C.Sturm_Waveform,2021TVT_T.Tian_Transmit,2021SJ_C.Shi_Joint,2021TSP_M.Keskin_Limited,2021ArXiv_Z.Xu_AWideband}. 
The orthogonal property of OFDM waveform fulfilled by the discrete Fourier transform (DFT) and inverse DFT (IDFT) operations at transceivers can facilitate signal processing for both communications and radar sensing. 
A pioneer work \cite{2011PIEEE_C.Sturm_Waveform} considered a communication-centric DFRC multiple-input multiple output (MIMO) OFDM system which provided a way of introducing radar sensing into a MIMO-OFDM system.
In \cite{2021TVT_T.Tian_Transmit}, a radar-centric DFRC MIMO-OFDM system with sidelobe control for communications was studied where the transmit and receive beamforming were optimized to maximize the radar metric of Kullback-Leibler divergence. The total radiated power was minimized in another radar-centric DFRC work \cite{2021SJ_C.Shi_Joint} through subcarrier selection and power allocation with a primary radar purpose and a secondary communications purpose.
Both radar-centric and joint design scenarios were addressed in recent work \cite{2021TSP_M.Keskin_Limited} where the DFRC trade-off waveform was optimized by taking the feedback overhead for conveying transmit waveform control information into consideration.
A wideband DFRC MIMO-OFDM waveform was jointly designed in \cite{2021ArXiv_Z.Xu_AWideband}, where precoding and antenna selection matrices were optimized to meet a joint communication-sensing performance.

Even though OFDM waveform is an excellent candidate for joint design of DFRC systems, one major disadvantage of OFDM waveform is high peak-to-average power ratio (PAPR) which should be effectively dealt with. 
Otherwise, high PAPR may cause non-linear distortion of the transmit signals and lead to radar/communication performance degradation considering the limited linear region of the low-cost power amplifiers \cite{WC2005_S.Han_AnOverview,2009JCM_D.Lim_AnOverview,2013SPM_G.Wunder_ThePAPR}.
In \cite{2021TVT_T.Tian_Transmit}, the PAPR constraints were considered in a radar-centric DFRC system to restrain the PAPR of MIMO-OFDM waveform.
However, to the best of our knowledge, most of the DFRC OFDM systems in the state-of-the-art literature have not taken the PAPR constraints into consideration, especially for the joint design DFRC systems. Moreover, PAPR is originally defined on the continuous-time passband signals \cite{B2010_Y.Cho_MIMOOFDM}, and thus oversampling may be necessary for obtaining accurate PAPR levels based on the discrete-time MIMO-OFDM waveform.
It is verified that the discrete-time OFDM signal can get almost the same PAPR as the continuous-time signal if it is $\Upsilon$-times interpolated (oversampled) with $\Upsilon\geq 4$ \cite{B2010_Y.Cho_MIMOOFDM,WC2005_S.Han_AnOverview}.


\vspace{-3mm}
\subsection{Our Contributions}\label{sec:Contributions}
In this paper, we consider a wideband  DFRC multi-user MIMO (MU-MIMO) OFDM system, where a multi-antenna DFRC base station (BS) acts as a cellular BS and a MIMO radar to generate DFRC MIMO-OFDM waveforms used for downlink data transmission and far-field targets detection  simultaneously.
We design low-PAPR DFRC MIMO-OFDM waveforms so as to achieve a tunable performance trade-off between wireless communications and radar sensing.

Our main contributions are summarized as follows:
\begin{itemize}
  \item \textbf{\em  We jointly design the DFRC MIMO-OFDM waveforms by optimizing a weighted objective function on both communications and radar performance metrics under transmit power and PAPR constraints.} The PAPR constraints are introduced to restrain the modulus of the DFRC MIMO-OFDM waveforms so as to avoid signal distortion caused by hardware limitation, e.g., low-cost non-linear power amplifiers, leading to energy-efficient data transmissions and target detections.

  \item \textbf{\em We first address low-PAPR DFRC MIMO-OFDM  waveform design in the scenario with Nyquist-rate sampling.} 
      The PAPR of discrete-time DFRC MIMO-OFDM waveform with Nyquist-rate sampling is derived and used for PAPR constraints. The formulated problem is transformed into a standard semi-definite programming (SDP) and can be optimally solved by the semi-definite relaxation (SDR) method, where we can prove that the rank-1 solution exists in general.

  \item \textbf{\em A low-complexity method is further proposed to divide the original problem into sub-problems and obtain an effective solution with much reduced complexity.} All the subproblems can be solved parallelly with much less optimization variables, and thus the computational complexity can be significantly reduced.

  \item \textbf{\em We further consider the practical case of DFRC MIMO-OFDM waveform with an oversampling rate $\Upsilon\geq 4$, and design a new low-PAPR approach.} Oversampling is required in practical OFDM systems with low-cost non-linear power amplifiers, where the PAPR level of ODFM waveform  can be measured in a more accurate way. We derive the PAPR expression and transform the optimization problem for low-PAPR DFRC MIMO-OFDM waveform design  into a standard SDP though elaborate transformations. The global optimal solution can also be achieved by the SDR method which is shown to offer rank-1 solution in general. We further advance this approach with a low-complexity method to obtain an effective solution.
\end{itemize}

The simulation results verify that the PAPR levels of our designed DFRC MIMO-OFDM waveform can be reduced from 12.5dB/10.5dB to 3dB in the case of radar/communication priority with very slight performance degradation.
Satisfactory performance tradeoff  between the communications and radar can be achieved through effectively tuning the weighting factor on these two functionalities, demonstrating the feasibility, effectiveness, and flexibility of the proposed methods for low-PAPR DFRC MIMO-OFDM waveform design.

The rest of this paper is organized as follows. Section \ref{sec:system} describes the considered system models including the DFRC waveform, the communication, and the radar models. The problem formulation as well as the low-PAPR DFRC MIMO-OFDM waveform design with Nyquist-rate sampling are given in Section \ref{sec:Problem_CP}. In Section \ref{sec:Problem_CPOS}, the practical scenario with oversampling is considered where the optimization problem is solved based on a more accurate PAPR measurement. The simulation results are provided in  Section~\ref{sec:simulation}, and we conclude our paper in Section~\ref{sec:conclusion}.

{\em Notations:} Unless  specified otherwise, the upper and lower case bold symbols represent matrices (i.e., $\mathbf{A}$) and vectors (i.e., $\mathbf{x}$). The notations $(\cdot)^\mathrm{*}$, $(\cdot)^\mathrm{T}$ and $(\cdot)^\mathrm{H}$ denote complex conjugate, transpose and Hermitian transpose operations for vectors or matrices. Also, $\mathbf{A}^\dag$ represents the Moore-Penrose pseudo-inverse of matrix $\mathbf{A}$. $\mathrm{Tr}\left\{\mathbf{X}\right\}$ is the trace of square matrix $\mathbf{X}$ and $\mathrm{rank}(\mathbf{A})$ represents the rank of matrix $\mathbf{A}$.  $[\mathbf{A}]_{i,j}$ denotes the $(i,j)$-th element of matrix $\mathbf{A}$.
$\mathrm{diag}(\cdot)$ can  form a diagonal-form matrix using elements in $(\cdot)$. $\mathrm{vec}(\mathbf{A})$ represents vectorization on matrix $\mathbf{A}$. $\otimes$ denotes the Kronecker product. $|\cdot|$, $\left\|\mathbf{x}\right\|$, $\left\|\mathbf{A}\right\|_\mathrm{F}$ denote the absolute value of scalers, the $l_2$ norm of vector $\mathbf{x}$ and Frobenius norm of matrix $\mathbf{A}$, respectively.


\begin{figure}[tbp]
  \centering
  \includegraphics[scale=0.39]{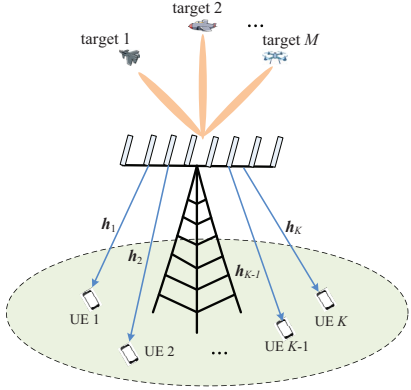}\\
  \vspace{-2mm}
  \caption{An illustration of the DFRC MIMO-OFDM  system, where a $N_\mathrm{t}$-antenna DFRC-BS simultaneously serves $K$ downlink UEs and probes $M$ far-field targets through jointly designing the DFRC MIMO-OFDM waveform.}\label{fig:system_model}
\end{figure}
\vspace{-2mm}
\section{System Model}\label{sec:system}
We consider a wideband DFRC  MU-MIMO system as shown in Fig. \ref{fig:system_model}, where a DFRC-BS equipped with a uniform linear array (ULA) of $N_\mathrm{t}$-antenna transmits DFRC waveforms, aiming at serving $K$ downlink single-antenna user equipment (UEs) and sensing $M$ far-field targets  simultaneously. 

In the considered wideband DFRC MIMO-OFDM system with sufficient bandwidth, high distance resolution can be achieved for radar detection. However, the communication channels with multiple taps undergo significant frequency-selective fading  due to the fact that  the bandwidth is much larger than the channel coherence bandwidth  \cite{B_D.Tse2005Fundamentals}.
It is well known that OFDM is an effective technique for overcoming frequency-selective fading while fully exploiting the frequency diversity, which is widely used in practical communication systems.
Through the OFDM technique, each wideband channel can be divided into multiple orthogonal frequency-flat subchannels \cite{B2010_Y.Cho_MIMOOFDM}.
On the other hand, the OFDM technique introduces the disadvantage of enlarging the PAPR of the transmitted signal waveform, which may lead to crucial signal distortion. Hence, the transmitted signal waveform should be carefully designed so as to effectively restrain the corresponding PAPR of the OFDM signals.

In this paper, the proposed system is denoted as DFRC MIMO-OFDM system, and the transmitted signal waveform is named as DFRC MIMO-OFDM waveform.
Our objective is to effectively design the transmitted DFRC MIMO-OFDM waveform under transmit power and specific PAPR constraints, so as to obtain desirable low-PAPR DFRC MIMO-OFDM waveform  achieving satisfactory performance tradeoff between communications and radar functionalities.

\subsection{DFRC MIMO-OFDM Waveform Formulation}\label{sec:DFRC_OFDM_Model}

We assume that the wireless communication channels are with the memory of $U-1$, i.e., each has $U$  effectively non-zero channel taps.
In addition, the length of the effective data symbols in each block is $N_\mathrm{s}$ which need modulate $N_\mathrm{s}$ subcarriers for transmissions \cite{B_D.Tse2005Fundamentals}, and we denote the index of symbols/subcarriers as $n \in \mathcal{N}_\mathrm{s}=\{1,\cdots, N_\mathrm{s}\}$.
In order to eliminate inter-symbol interference (ISI) of wideband multicarrier transmissions, the standard OFDM technique of  cyclic prefix (CP)  is utilized \cite{B2010_Y.Cho_MIMOOFDM}. 
The length of CP is denoted as $N_\mathrm{c}$ with $N_\mathrm{c}\geq U-1$, and thus the total number of time-domain samples per block is $N=N_\mathrm{s}+N_\mathrm{c}$. Without lose of generality, we use $N_\mathrm{c}= U-1$ in this paper.



\subsubsection{Symbol Data and Preoding Model}
Let 
\begin{align}
\mathbf{S}=[\mathbf{S}{_1^\mathrm{T}},\cdots,\mathbf{S}{_{N_\mathrm{s}}^\mathrm{T}}]^\mathrm{T} \in \mathbb{C}^{N_\mathrm{s}K\times L}
\end{align}
represent the symbol data matrix for all UEs of $k \in \mathcal{K}=\{1,\cdots,K\}$ on  $N_\mathrm{s}$ subcarriers,  transmitted  during a communication frame with length  $L$.
\footnote{$L$ is the length of the radar pulse/communication block, corresponding to the number of OFDM symbols in the time domain.}
Here, $\mathbf{S}_n=[\mathbf{s}_{n,1},\cdots,\mathbf{s}_{n,K}]^\mathrm{T} \in \mathbb{C}^{K\times L}$ is the symbol matrix for all UEs on subcarrier  $n \in \mathcal{N}_\mathrm{s}$ with $\mathbf{s}_{n,k} \in \mathbb{C}^{L\times 1}$ being the specific symbol vector for user $k$. 
In addition, $\mathbf{W}=\mathrm{diag}(\mathbf{W}_1,\cdots,\mathbf{W}_{N_\mathrm{s}}) \in \mathbb{C}^{N_\mathrm{s}N_\mathrm{t}\times N_\mathrm{s}K}$ is the compact precoding matrix for all UEs on all  subcarriers  where $\mathbf{W}_n=[\mathbf{w}_{n,1},\cdots,\mathbf{w}_{n,K}] \in \mathbb{C}^{N_\mathrm{t}\times K}$ is the precoding matrix for all UEs on subcarrier $n \in \mathcal{N}_\mathrm{s}$ with $\mathbf{w}_{n,k} \in \mathbb{C}^{N_\mathrm{t}\times 1}$ being the specific precoding vector for user $k$. 

The transmit symbol data $\mathbf{S}$ is first precoded by $\mathbf{W}$ in the frequency domain and then converted to the time domain by IDFT operation before subcarrier modulations. 
We further use 
\begin{align}
\mathbf{X}_\mathrm{s}=\mathbf{W}\mathbf{S}=[\mathbf{X}_1^\mathrm{T},\cdots, \mathbf{X}_{N_\mathrm{s}}^\mathrm{T}]^\mathrm{T} \in \mathbb{C}^{N_\mathrm{s}N_\mathrm{t}\times L}
\end{align}
to indicate the baseband precoded symbol matrix on all subcarriers before the IDFT processing where $\mathbf{X}_n=\mathbf{W}_n\mathbf{S}_n \in \mathbb{C}^{N_\mathrm{t}\times L}$ for $n \in \mathcal{N}_\mathrm{s}$.
We denote $\mathbf{F}_\mathrm{s} \in \mathbb{C}^{N_\mathrm{s}\times N_\mathrm{s}}$ as the normalized DFT matrix for data transmissions with 
\begin{align}
F_{n,m}^\mathrm{s}=\frac{1}{\sqrt{N_\mathrm{s}}}e^{-\frac{j2\pi}{N_\mathrm{s}}(n-1)(m-1)}
\end{align}
being the ($n,m$)-th element of $\mathbf{F}_\mathrm{s}$ for $n,m \in \mathcal{N}_\mathrm{s}$.
The IDFT processing at the transmitter of the DFRC-BS is operated by
$(\mathbf{F}_\mathrm{s}^\mathrm{H}\otimes \mathbf{I}_{N_\mathrm{t}}) \in \mathbb{C}^{ N_\mathrm{s}N_\mathrm{t}\times N_\mathrm{s}N_\mathrm{t}}$ 
considering the fact that the ULA array is equipped with $N_\mathrm{t}$ transmit antennas.
Hence, the transmitted time-domain data signal waveform after IDFT can be further expressed as
\begin{align}\label{Gs_sWaveform}
\hspace{-2mm}\mathbf{G}&\triangleq(\mathbf{F}_\mathrm{s}^\mathrm{H}\otimes \mathbf{I}_{N_\mathrm{t}})\mathbf{X}_\mathrm{s} =[\mathbf{G}_{1}^\mathrm{T},\cdots, \mathbf{G}_{N_\mathrm{s}}^\mathrm{T}]^\mathrm{T} \in \mathbb{C}^{N_\mathrm{s}N_\mathrm{t}\times L},
\end{align}
where $\mathbf{G}_{n}=(\mathbf{f}_{\mathrm{s},n}^\mathrm{H}\otimes \mathbf{I}_{N_\mathrm{t}})\mathbf{X}_\mathrm{s} \in \mathbb{C}^{N_\mathrm{t}\times L}$ is the transmitted DFRC MIMO-OFDM signal from $N_\mathrm{t}$ antennas on the $n$-th subcarrier for $n\in\mathcal{N}_s$.
Also, $\mathbf{f}_{\mathrm{s},n} \in \mathbb{C}^{N_\mathrm{s}\times 1}$ indicates the $n$-th column of the DFT matrix $\mathbf{F}_\mathrm{s}$ corresponding to the (I)DFT operations on the $n$-th subcarrier.

\subsubsection{DFRC MIMO-OFDM Waveform with CP}
The next step is adding CP with size $N_\mathrm{c}= U-1$, which is crucial for eliminating the ISI caused by multi-path frequency-selective fading.
It is operated by repeating  the last $N_\mathrm{c}$ symbols at the beginning of the original symbol sequence \cite{TCOM2002_B.Muquet_Cyclic}.
We denote 
\begin{align}
&\hspace{-1mm}\mathbf{X}_\mathrm{c}=[\mathbf{X}_{N_\mathrm{s}-N_\mathrm{c}+1}^\mathrm{T},\mathbf{X}_{N_\mathrm{s}-N_\mathrm{c}+2}^\mathrm{T},\cdots, \mathbf{X}_{N_\mathrm{s}}^\mathrm{T}]^\mathrm{T} \in \mathbb{C}^{N_\mathrm{c}N_\mathrm{t}\times L}, \label{Xc}\\
&\hspace{-1mm}\mathbf{F}_\mathrm{c}=[\mathbf{f}_{\mathrm{s},N_\mathrm{s}-N_\mathrm{c}+1}, \mathbf{f}_{\mathrm{s},N_\mathrm{s}-N_\mathrm{c}+2},\cdots, \mathbf{f}_{\mathrm{s},N_\mathrm{s}}] \in \mathbb{C}^{N_\mathrm{s}\times N_\mathrm{c}}, \label{Fc}
\end{align}
and let  
\begin{align}
&\dot{\mathbf{X}}=[\mathbf{X}_\mathrm{c}^\mathrm{T}, \mathbf{X}_\mathrm{s}^\mathrm{T}]^\mathrm{T} \in \mathbb{C}^{NN_\mathrm{t}\times L},\\
&\dot{\mathbf{F}}=[\mathbf{F}_\mathrm{c}, \mathbf{F}_\mathrm{s}]\in \mathbb{C}^{N_\mathrm{s}\times N}, 
\end{align}
then the transmitted DFRC MIMO-OFDM waveform after adding CP can be expressed as 
\begin{align}
\hspace{-2mm}&\dot{\mathbf{G}}\triangleq(\dot{\mathbf{F}}^\mathrm{H}\otimes \mathbf{I}_{N_\mathrm{t}})\dot{\mathbf{X}} \in \mathbb{C}^{NN_\mathrm{t}\times L} \nonumber\\
&=[\mathbf{G}_{N_\mathrm{s}-N_\mathrm{c}+1}^\mathrm{T},\mathbf{G}_{N_\mathrm{s}-N_\mathrm{c}+2}^\mathrm{T},\cdots, \mathbf{G}_{N_\mathrm{s}}^\mathrm{T},\mathbf{G}_{1}^\mathrm{T},\cdots, \mathbf{G}_{N_\mathrm{s}}^\mathrm{T}]^\mathrm{T},
\end{align}
recalling that $N=N_\mathrm{c}+N_\mathrm{s}$.
Through operating the IDFT and adding the CP at the transmitter side, the effects of ISI caused by frequency-selective fading can be eliminated after removing the CP and operating the DFT at the receiver side.


\subsection{MIMO-OFDM Communication Model}\label{sec:CommsModel}
\subsubsection{Overall Downlink Channel and Received Signal}
After adding the CP, the signal is then up-converted to the radio frequency (RF) domain through subcarrier modulations for transmission via $N_\mathrm{t}$ RF chains connected to $N_\mathrm{t}$ antennas. 

As mentioned before, the wireless channels from the DFRC-BS to the downlink UEs are assumed as wideband frequency-selective fading channels with memory of $U-1$, which is assumed time-invariant for a sufficiently long time-scale \cite{B_D.Tse2005Fundamentals}. 
Hence, the discrete-time impulse response of the channel from the DFRC-BS to UE $k \in \mathcal{K}$  can be denoted as
\begin{align}
{\widetilde{\mathbf{h}}}_k(\tau)=\sum_{u=0}^{U-1} \widetilde{\mathbf{h}}_{k,u}\delta(\tau-u)\in \mathbb{C}^{1\times N_\mathrm{t}}, \ \forall k \in \mathcal{K},\ \forall \tau
\end{align}
where $\widetilde{\mathbf{h}}_{k,u}\in \mathbb{C}^{1\times N_\mathrm{t}}$ is the time-invariant channel vector of the $u$-th tap   and it is assumed as $\widetilde{\mathbf{h}}_{k,u}\sim\mathcal{CN}(0,\frac{1}{U}\mathbf{I}_{N_\mathrm{t}})$ with independent identically distributed (i.i.d.)  Rayleigh fading coefficients for any $k \in \mathcal{K}$ and $u \in \mathcal{U}=\{0,\cdots, U-1\}$.
Furthermore, we use 
\begin{align} \label{Channel_Compact_T}
{\widetilde{\mathbf{H}}}(\tau)=\sum_{u=0}^{U-1} \widetilde{\mathbf{H}}_{u}\delta(\tau-u) \in \mathbb{C}^{K\times N_\mathrm{t}},\ \forall \tau
\end{align}
to represent the compact channel matrix from the DFRC-BS to $K$ downlink UEs, where
$\widetilde{\mathbf{H}}_{u}=[{\widetilde{\mathbf{h}}}_{1,u}^\mathrm{T}, \cdots, {\widetilde{\mathbf{h}}}_{K,u}^\mathrm{T}]^\mathrm{T} \in \mathbb{C}^{K\times N_\mathrm{t}}$ is the corresponding channel matrix of the $u$-th tap.

\begin{lemma}\label{lemma0}
The overall effective downlink channel matrix of the DFRC MIMO-OFDM  system on the $n$-th symbol subcarrier can be written in  frequency domain as \cite{B_D.Tse2005Fundamentals,TCOM2002_B.Muquet_Cyclic}
\begin{align} \label{fre_sele_channel2}
\mathbf{H}_n=\sum_{u=0}^{U-1} \widetilde{\mathbf{H}}_{u}e^{-\frac{j2\pi u(n-1)}{N_\mathrm{s}}}, \ \forall n \in \mathcal{N}_\mathrm{s}.
\end{align}
Also, we have $\mathbf{H}_n=[{\mathbf{h}}_{1,n}^\mathrm{T}, \cdots, {\mathbf{h}}_{K,n}^\mathrm{T}]^\mathrm{T} \in \mathbb{C}^{K\times N_\mathrm{t}}$ with
${\mathbf{h}}_{k,n}=\sum_{u=0}^{U-1} \widetilde{\mathbf{h}}_{k,u}e^{-\frac{j2\pi u(n-1)}{N_\mathrm{s}}} \in \mathbb{C}^{1\times N_\mathrm{t}}$ being the corresponding channel for UE $k \in \mathcal{K}$. 
In addition, the  received signal of symbol data sequence can be expressed as 
\begin{align}\label{Ys_compact}
\mathbf{Y}_\mathrm{s}=\mathbf{H}_\mathrm{s}\mathbf{X}_\mathrm{s}+\mathbf{Z}_\mathrm{s} \in \mathbb{C}^{N_\mathrm{s}K\times L}
\end{align}
where $\mathbf{H}_\mathrm{s}=\mathrm{diag}(\mathbf{H}_1,\cdots, \mathbf{H}_{N_\mathrm{s}}) \in \mathbb{C}^{N_\mathrm{s}K\times N_\mathrm{s}N_\mathrm{t}}$ and $\mathbf{Z}_\mathrm{s} \in \mathbb{C}^{N_\mathrm{s}K\times L}$ is the matrix of the additive white gaussian noise (AWGN) on $N_\mathrm{s}$ subcarriers with i.i.d.  random variables following $\mathcal{CN}(0,\sigma^2)$.
Specifically, the corresponding noiseless received signal on the subcarrier $n$  can be expressed as $\mathbf{H}_{n}\mathbf{X}_{n}$ for $n \in \mathcal{N}_\mathrm{s}$ \cite{B_D.Tse2005Fundamentals,TCOM2002_B.Muquet_Cyclic}.
\end{lemma}
\begin{proof}
See Appendix~\ref{sec:proof_lem0}.
\end{proof}

In this paper, we assume that the channels are known to the DFRC-BS, which can be obtained through conventional channel estimation methods.

\subsubsection{Multi-user Interference (MUI) and Sum Rate}
The downlink MU-MIMO OFDM transmissions will induce MUI, which will have great effects on the performance of the achievable sum rate as well as the symbol error rate (SER) of the downlink UEs. Assuming that the data symbols for all UEs of $k\in\mathcal{K}$ on all subcarriers of $n\in\mathcal{N}_\mathrm{s}$ follow the same constellation modulation, we can re-express the compact received signal in \eqref{Ys_compact} as
\begin{align}
\mathbf{Y}_\mathrm{s}=\underset{\mathrm{Signal}}{\underbrace{\mathbf{S}}}+\underset{\mathrm{MUI}}
{\underbrace{(\mathbf{H}_\mathrm{s}\mathbf{X}_\mathrm{s}-\mathbf{S})}}+\underset{\mathrm{Noise}}{\underbrace{\mathbf{Z}_\mathrm{s}}},
\end{align}
where $\mathbf{H}_\mathrm{s}\mathbf{X}_\mathrm{s}-\mathbf{S}$ represents the MUI signals caused by multi-user transmissions \cite{TCOM2013_S.Mohammed_PerAntenna}.
The signal-to-interference-plus-noise ratio (SINR) of UE $k\in\mathcal{K}$ on subcarrier $n\in\mathcal{N}_\mathrm{s}$ per frame can be further expressed as
\begin{align}\label{SINRkn}
\mathrm{SINR}_{k,n}=\frac{\mathbb{E}\{|s_{n,k}^l|^2\}}{ \mathbb{E}\{|{\mathbf{h}}_{k,n}\mathbf{x}_n^l-s_{n,k}^l|^2\}+\sigma^2},
\end{align}
where $s_{n,k}^l$ is the $l$-th element of the symbol vector $\mathbf{s}_{n,k}$ and $\mathbf{x}_n^l \in \mathbb{C}^{N_\mathrm{t}\times 1}$ is the $l$-th column of the precoded symbol matrix $\mathbf{X}_n$ for all UEs.
The expectations in  \eqref{SINRkn} are taken with respect to (w.r.t.) the time index  $l \in \mathcal{L}=\{1,\cdots,L\}$  \cite{TSP2018_F.Liu_Toward}, and $\mathbb{E}\{|s_{n,k}^l|^2\}$ is fixed with given constellation mode. Hence, the maximum achievable sum rate of the $K$ downlink UEs on subcarrier $n$ can be given as 
\vspace{-2mm}
\begin{align}
R_{n}=\sum_{n=1}^{K}\log_2\left(1+\mathrm{SINR}_{k,n}\right).
\end{align}
The average sum rate/spectral efficiency of the  MIMO-OFDM system is measured as 
\vspace{-2mm}
\begin{align}
R=\frac{1}{N_\mathrm{s}}\sum_{n=1}^{N_\mathrm{s}}R_{n}.
\end{align}

Note that the MUI is an important performance metric of the DFRC MIMO-OFDM system which should be minimized from the perspective of communications, so as to increase the achievable sum rate and decrease the SER. To this end, in the following we employ the MUI as our objective/cost function for enhancing the communication performance of the DFRC MIMO-OFDM system, given as
\begin{align}
\underset{\mathbf{X}_\mathrm{s}}{\min} ~\big\| \mathbf{H}_\mathrm{s}\mathbf{X}_\mathrm{s}-\mathbf{S} \big\|_\mathrm{F}^2.
\end{align}

\vspace{-4mm}
\subsection{MIMO-OFDM Radar Model}\label{sec:RadarModel}
\subsubsection{Radar Beampattern}
Radar beampattern is a crucial indicator for measuring the radar detection and tracking performance. Note that MIMO radar is capable of achieving higher degrees of freedom (DoFs) than traditional phased-array radar through generating uncorrelated MIMO waveforms \cite{B2008_J.Li_MIMO}.  It has been verified that  designing  MIMO radar beampattern can be equivalently transformed into designing the probing signal waveform as the  radar beampattern is highly related to the covariance matrix of the probing waveform \cite{TSP2018_F.Liu_Toward,B2008_J.Li_MIMO,2007TSP_P.Stoica_OnProbing,TAES2008_D.Fuhrmann_Transmit}. 
For the considered DFRC MIMO-OFDM system, the transmit radar beampattern versus the detection angle $\theta$ is written as
\begin{align}
B_\mathrm{d}(\theta)=\frac{1}{N_\mathrm{s}}\sum_{n=1}^{N_\mathrm{s}} \mathbf{a}^\mathrm{H}(\theta)\mathbf{R}_{\mathrm{G},n}\mathbf{a}(\theta),
\end{align}
which is averaged on the $N_\mathrm{s}$ subcarriers.
Here $\mathbf{a}(\theta) \in \mathbb{C}^{N_\mathrm{t}\times1 }$ is the transmit steering vector given as
\begin{align} 
\mathbf{a}(\theta) =[e^{-j\frac{N_\mathrm{t}-1}{2}\pi\sin\theta},e^{-j\frac{N_\mathrm{t}-3}{2}\pi\sin\theta},\cdots,e^{j\frac{N_\mathrm{t}+1}{2}\pi\sin\theta}]^\mathrm{T},
\end{align}
under the assumption that even number of transmit antennas are equipped at the DFRC-BS ULA and the center of the ULA is chosen as the reference point.
In addition, $\mathbf{R}_{\mathrm{G},n} \in \mathbb{C}^{N_\mathrm{t}\times N_\mathrm{t}}$ is the effective spatial covariance matrix of the transmit DFRC MIMO-OFDM waveform on the $n$-th subcarrier for $n \in \mathcal{N}_\mathrm{s}$ which is defined as
\begin{align}
\mathbf{R}_{\mathrm{G},n}=\frac{1}{L}\mathbf{G}_{n}\mathbf{G}_{n}^\mathrm{H}.
\end{align}
In order to ensure that the covariance matrices $\{\mathbf{R}_{\mathrm{G},n}\}_{n\in\mathcal{N}_\mathrm{s}}$ are positive-definite, we assume that the frame length satisfies $L\geq N_\mathrm{t}$, which is easy to achieve in the considered wideband scenario with approximately time-invariant channels.

\subsubsection{Radar Detection Probability}\label{Pd_SNRR}
From the perspective of radar, another important performance indicator is the detection probability.
To derive the detection probability, we first express the  radar received target echo signal  by the DFRC-BS at the $l$-th snapshot/frame as
\begin{align}
\mathbf{y}_l^\mathrm{r}=\alpha\mathbf{\Upsilon}(\theta)\mathbf{g}_l^\mathrm{s}+\mathbf{z}^\mathrm{r} \in \mathbb{C}^{N_\mathrm{s}N_\mathrm{t}\times 1},
\end{align}
which is an expanded radar received echo vector considering the $N_\mathrm{s}$ i.i.d data streams on the $N_\mathrm{s}$ subcarriers for $l \in \mathcal{L}$. 
Here, $\alpha$ is the complex path loss of the radar-target-radar path, $\mathbf{g}_l^\mathrm{s}=(\mathbf{F}_\mathrm{s}^\mathrm{H}\otimes \mathbf{I}_{N_\mathrm{t}})\mathbf{x}_\mathrm{s}^l \in \mathbb{C}^{N_\mathrm{s}N_\mathrm{t}\times 1}$ with $\mathbf{x}_\mathrm{s}^l \in \mathbb{C}^{N_\mathrm{s}N_\mathrm{t}\times 1}$ being the $l$-th column of the precoded signal matrix $\mathbf{X}_\mathrm{s}$, and $\mathbf{z}^\mathrm{r}\sim \mathcal{CN}(0,\sigma_\mathrm{r}^2\mathbf{I}_{N_\mathrm{s}N_\mathrm{t}})$ is the AWGN noise for radar reception.
In addition, $\mathbf{\Upsilon}(\theta)=\mathbf{I}_{N_\mathrm{s}}\otimes \widetilde{\mathbf{\Upsilon}}(\theta) \in \mathbb{C}^{N_\mathrm{s}N_\mathrm{t}\times N_\mathrm{s}N_\mathrm{t}}$ with $\widetilde{\mathbf{\Upsilon}}(\theta)= \mathbf{a}_\mathrm{r}(\theta)\mathbf{a}_\mathrm{t}^\mathrm{T}(\theta)\in \mathbb{C}^{N_\mathrm{t}\times N_\mathrm{t}}$ where $\mathbf{a}_\mathrm{r}(\theta)$ and $\mathbf{a}_\mathrm{t}(\theta)$ are the transmit and receive steering vector with $\mathbf{a}_\mathrm{r}(\theta)=\mathbf{a}_\mathrm{t}(\theta)=\mathbf{a}(\theta)$.

Next by leveraging the Generalized Likelihood Ratio Test (GLRT) \cite{TSP2018_F.Liu_Toward,TSP2006_I.Bekkerman_Target,PHP1998_S.Kay_Fundamentals}, we can further obtain the asymptotic radar detection probability as
\begin{align}\label{Pd_Pf}
\mathcal{P}_\mathrm{D}=1-\mathfrak{F}_{\mathcal{X}_2^2\left(\mu\right)}(\zeta)=1-\mathfrak{F}_{\mathcal{X}_2^2(\mu)}\left( \mathfrak{F}_{\mathcal{X}_2^2}^{-1}(1-\mathcal{P}_\mathrm{f}) \right),
\end{align}
where $P_\mathrm{f}=1-\mathfrak{F}_{\mathcal{X}_2^2}(\zeta)$ is the false alarm rate and $\mathfrak{F}_{\mathcal{X}_2^2}$ is the chi-squared cumulative distribution function (CDF) with 2 DoFs. In order to endure a constant false alarm rate $\mathcal{P}_\mathrm{f}$, we have $\zeta=\mathfrak{F}_{\mathcal{X}_2^2}^{-1}(1-\mathcal{P}_\mathrm{f})$ under the Neyman-Pearson criterion with $\mathfrak{F}_{\mathcal{X}_2^2}^{-1}$ being the inverse function of the chi-squared CDF with 2 DoFs. In addition, the  function $\mathfrak{F}_{\mathcal{X}_2^2(\mu)}$ in \eqref{Pd_Pf} is the non-central chi-squared CDF with 2 DoFs and the non-central parameter $\mu$ is defined as
\begin{align}
\mu&=L|\alpha|^2P_\mathrm{t}^\mathrm{r}\mathrm{tr}\left( \mathbf{\Upsilon}(\theta)\mathbf{R}_\mathrm{G}\mathbf{\Upsilon}^\mathrm{H}(\theta)(\sigma_\mathrm{r}^2\mathbf{I})^{-1}\right) \nonumber \\
&=\mathrm{SNR}_\mathrm{R}\mathrm{tr}\left(  \mathbf{\Upsilon}(\theta)\mathbf{R}_\mathrm{G}\mathbf{\Upsilon}^\mathrm{H}(\theta)\right),
\end{align}
where $P_\mathrm{t}^\mathrm{r}$ is the power of the DFRC MIMO-OFDM probing waveform and $\mathbf{R}_\mathrm{G}=\frac{1}{L}\mathbf{G}\mathbf{G}^\mathrm{H} \in \mathbb{C}^{N_\mathrm{s}N_\mathrm{t}\times N_\mathrm{s}N_\mathrm{t}}$.
In addition, the signal-to-noise (SNR) ratio of the radar received  target echo signal  is denoted as $\mathrm{SNR}_\mathrm{R}=\frac{L|\alpha|^2P_\mathrm{t}}{\sigma_\mathrm{r}^2}$ \cite{TSP2006_I.Bekkerman_Target}. 


To design our DFRC MIMO-OFDM waveform, we employ a desired benchmark of the MIMO-OFDM radar waveform, denoted as $\mathbf{G}_0\triangleq(\mathbf{F}_\mathrm{s}^\mathrm{H}\otimes \mathbf{I}_{N_\mathrm{t}})\mathbf{X}_{\mathrm{s},0}$ with $\mathbf{X}_{\mathrm{s},0}$ being the radar waveform before IDFT operation, which is capable to achieve desirable detection probability. One such benchmark waveform of $\mathbf{G}_0$ can be obtained by leveraging the \emph{Directional Beampattern Design}.
From the viewpoint of radar performance, it is beneficial to make the DFRC waveform $\mathbf{G}=(\mathbf{F}_\mathrm{s}^\mathrm{H}\otimes \mathbf{I}_{N_\mathrm{t}})\mathbf{X}_\mathrm{s}$ to be as close to $\mathbf{G}_0$ as possible. Hence, the objective/cost function for enhancing the radar performance of the DFRC MIMO-OFDM system is given as
\begin{align}
\underset{\mathbf{X}_\mathrm{s}}{\min} ~\big\| (\mathbf{F}_\mathrm{s}^\mathrm{H}\otimes \mathbf{I}_{N_\mathrm{t}})\mathbf{X}_\mathrm{s}-\mathbf{G}_0 \big\|_\mathrm{F}^2.
\end{align}

\vspace{-4mm}
\section{Low-PAPR DFRC MIMO-OFDM Waveform Design with Nyquist-Rate Sampling}\label{sec:Problem_CP}
In this section, we consider the scenario leveraging the Nyquist-rate sampling  where the sampled sequence corresponds to the symbol sequence \cite{2000TIT_K.Paterson_OntheExistence}.
Our aim is to achieve a desirable  performance tradeoff between communications and radar detections through elaborately designing the precoded DFRC MIMO-OFDM waveform matrix $\mathbf{X}_\mathrm{s} \in \mathbb{C}^{N_\mathrm{s}N_\mathrm{t}\times L}$ under the PAPR as well as the power allocation constraints.\footnote{It is easy to notify that the signal matrix with CP, i.e., $\dot{\mathbf{X}}$, can be totally determined by the symbol signal matrix $\mathbf{X}_\mathrm{s}$ according to their definitions in section \ref{sec:DFRC_OFDM_Model}.}

\vspace{-4mm}
\subsection{PAPR of the DFRC MIMO-OFDM Waveform}\label{sec:PAPR}
PAPR is defined as the ratio between the maximum power and the average power of the complex passband signal  \cite{B2010_Y.Cho_MIMOOFDM}, and thus the PAPR constraint on the considered discrete-time DFRC MIMO-OFDM waveform can be given as
\begin{align}\label{PAPR_C}
\mathrm{PAPR}(\mathbf{X}_\mathrm{s})=\frac{\underset{i,l}{\max} ~\Big|\left[\big(\dot{\mathbf{F}}^\mathrm{H}\otimes \mathbf{I}_{N_\mathrm{t}}\big)\dot{\mathbf{X}}\right]_{i,l}\Big|^2}
{\frac{1}{NN_\mathrm{t}L}\left\|\big(\dot{\mathbf{F}}^\mathrm{H}\otimes \mathbf{I}_{N_\mathrm{t}}\big)\dot{\mathbf{X}} \right\|_\mathrm{F}^2}\leq\varepsilon, 
\end{align}
where $\varepsilon\in[1,NN_\mathrm{t}L]$ is the maximum allowable PAPR threshold of the DFRC MIMO-OFDM  system.
Based on the fact that the radar is usually required to transmit at its maximum available power in practice, we then have the following equality power allocation constraint
\begin{align} \label{PA_C}
\frac{1}{L}\left\|\big(\dot{\mathbf{F}}^\mathrm{H}\otimes \mathbf{I}_{N_\mathrm{t}}\big)\dot{\mathbf{X}}\right\|_\mathrm{F}^2= P_\mathrm{t},
\end{align}
where the power budget $P_\mathrm{t}$ is totally utilized for generating the DFRC MIMO-OFDM waveform.
Hence, this PAPR constraint defined in \eqref{PAPR_C} can be simplified as
\begin{align}\label{PAPR_C1}
&\underset{i,l}{\max} ~\Big|\left[\big(\dot{\mathbf{F}}^\mathrm{H}\otimes \mathbf{I}_{N_\mathrm{t}}\big)\dot{\mathbf{X}}\right]_{i,l}\Big|^2\leq \frac{\varepsilon P_\mathrm{t}}{NN_\mathrm{t}},  \nonumber \\
&~\forall i\in\dot{\mathcal{N}}_\mathrm{t}=\{1,\cdots, NN_\mathrm{t}\}, ~\forall l\in\mathcal{L},
\end{align}
which is equivalent to the following set of constraints
\begin{align}\label{PAPR_C2}
&\left|\Big[\left(\mathbf{F}_\mathrm{s}^\mathrm{H}\otimes \mathbf{I}_{N_\mathrm{t}}\right)\mathbf{X}_\mathrm{s}\Big]_{i,l}\right|^2\leq \frac{\varepsilon P_\mathrm{t}}{NN_\mathrm{t}}, \nonumber \\
&~\forall i\in\mathcal{N}_\mathrm{st}=\{1,\cdots, N_\mathrm{s}N_\mathrm{t}\}, ~\forall l\in\mathcal{L}.
\end{align}

\subsection{Problem Formulation}\label{sec:PF_CP}
In order to achieve desirable performance tradeoff between communications and radar under the PAPR and power allocation constraints,  we can formulate the DFRC MIMO-OFDM waveform optimization problem as
\begin{subeqnarray}\label{eq:OFDM3}
\hspace{-6mm}({\rm P1}): \underset{ \mathbf{X}_\mathrm{s} }{\min} &&\hspace{-4mm} \frac{\rho}{\left\|\mathbf{S}\right\|_\mathrm{F}^2}
\big\| \mathbf{H}_\mathrm{s}\mathbf{X}_\mathrm{s}-\mathbf{S} \big\|_\mathrm{F}^2 + \nonumber \\
&&\hspace{-4mm} \frac{1-\rho}{\|\mathbf{G}_{0}\|_\mathrm{F}^2}
\big\|(\mathbf{F}_\mathrm{s}^\mathrm{H}\otimes \mathbf{I}_{N_\mathrm{t}})\mathbf{X}_\mathrm{s}-\mathbf{G}_0 \big\|_\mathrm{F}^2 \slabel{eq:OFDM3_0} \\
\mathrm{s.t.}
&&\hspace{-4mm} \left|\Big[\left(\mathbf{F}_\mathrm{s}^\mathrm{H}\otimes \mathbf{I}_{N_\mathrm{t}}\right)\mathbf{X}_\mathrm{s}\Big]_{i,l}\right|^2\leq \frac{\varepsilon P_\mathrm{t}}{NN_\mathrm{\mathrm{t}}}, \nonumber \\
&&\hspace{-4mm}  ~\forall i\in\mathcal{N}_\mathrm{st}, ~\forall l\in\mathcal{L}, \slabel{eq:OFDM3_1}\\
&&\hspace{-4mm} \frac{1}{L}\left\|\big(\dot{\mathbf{F}}^\mathrm{H}\otimes \mathbf{I}_{N_\mathrm{t}}\big)\dot{\mathbf{X}}\right\|_\mathrm{F}^2= P_\mathrm{t}, \slabel{eq:OFDM3_2}
\end{subeqnarray} 
where a weighted objective function considering the normalized objectives relating to the performance of communications and radar is leveraged.
It is easy to verify that
\begin{align}
&\left\|\big(\dot{\mathbf{F}}^\mathrm{H}\otimes \mathbf{I}_{N_\mathrm{t}}\big)\dot{\mathbf{X}}\right\|_\mathrm{F}^2 =\left\|\dot{\mathbf{G}} \right\|_\mathrm{F}^2  \nonumber \\
=&\left\| \left(\mathbf{F}_\mathrm{s}^\mathrm{H}\otimes \mathbf{I}_{N_\mathrm{T}} \right)\mathbf{X}_\mathrm{s} \right\|_\mathrm{F}^2+
\left\|\mathbf{\Gamma}_\mathrm{c}\left(\mathbf{F}_\mathrm{s}^\mathrm{H}\otimes \mathbf{I}_{N_\mathrm{T}} \right)\mathbf{X}_\mathrm{s} \right\|_\mathrm{F}^2 \nonumber \\
=&\left\| \mathbf{G}\right\|_\mathrm{F}^2 +\left\| \mathbf{\Gamma}_\mathrm{c}\mathbf{G}\right\|_\mathrm{F}^2
\end{align}
where $\mathbf{\Gamma}_\mathrm{c} \in \mathbb{R}^{N_\mathrm{s}N_\mathrm{t}\times N_\mathrm{s}N_\mathrm{t}}$ defined as
\begin{align}\label{Gcp_G}
\mathbf{\Gamma}_\mathrm{c}=
\left[
 \begin{matrix}
   \mathbf{0}_{(N_\mathrm{s}-N_\mathrm{c})N_\mathrm{t}}   \hspace{-4mm}&\mathbf{0}_{ (N_\mathrm{s}-N_\mathrm{c})N_\mathrm{t}\times N_\mathrm{c}N_\mathrm{t} } \\
   \mathbf{0}_{N_\mathrm{c}N_\mathrm{t} \times(N_\mathrm{s}-N_\mathrm{c})N_\mathrm{t}}    \hspace{-4mm}&\mathbf{I}_{N_\mathrm{c}N_\mathrm{t} }
  \end{matrix}
\right]
\end{align}
is utilized to abstract the signals in $\mathbf{G}$ used for CP that
\begin{align}\label{GaG}
\hspace{-2mm} \mathbf{\Gamma}_\mathrm{c}\mathbf{G}=[\mathbf{0}, \cdots, \mathbf{0}, \mathbf{G}_{N_\mathrm{s}-N_\mathrm{c}+1}^\mathrm{T},\cdots, \mathbf{G}_{N_\mathrm{s}}^\mathrm{T}]^\mathrm{T} \in \mathbb{C}^{N_\mathrm{s}N_\mathrm{t}\times L}
\end{align}
where the zero matrix $\mathbf{0}$ in \eqref{GaG} is with the size of $L \times N_\mathrm{t}$. 
In order to facilitate the solving process of problem (P1) in \eqref{eq:OFDM3}, we can divide the equality power allocation constraint \eqref{eq:OFDM3_2} into two equality power allocation constraints respectively on the symbol signals and the CP signals as
\begin{align}
&\frac{1}{L} \left\| \left(\mathbf{F}_\mathrm{s}^\mathrm{H}\otimes \mathbf{I}_{N_\mathrm{T}} \right)\mathbf{X}_\mathrm{s} \right\|_\mathrm{F}^2
=\frac{1}{L} \left\| \mathbf{G}\right\|_\mathrm{F}^2 = \beta P_\mathrm{t}, \label{PC_S}\\
&\frac{1}{L}\left\|\mathbf{\Gamma}_\mathrm{c}\left(\mathbf{F}_\mathrm{s}^\mathrm{H}\otimes \mathbf{I}_{N_\mathrm{T}} \right)\mathbf{X}_\mathrm{s} \right\|_\mathrm{F}^2
= \frac{1}{L} \left\| \mathbf{\Gamma}_\mathrm{c}\mathbf{G}\right\|_\mathrm{F}^2 = (1-\beta )P_\mathrm{t}, \label{PC_C1}
\end{align}
by introducing a power splitting parameter $\beta \in (0,1)$.
For simplicity, we fix $\beta=\frac{N_\mathrm{s}}{N}$ and then we have $\beta P_\mathrm{t}=\frac{N_\mathrm{s}}{N}P_\mathrm{t}\triangleq P_\mathrm{t}^\mathrm{s}$ and $(1-\beta)P_\mathrm{t}=\frac{N_\mathrm{c}}{N}P_\mathrm{t}\triangleq P_\mathrm{t}^\mathrm{c}$, which is reasonable on the basis of average power allocation on data transmissions and will be used in the rest of this paper.\footnote{The technique of CP is beneficial to avoid ISI cased  by frequency-selective fading but at the cost of consuming extra energy and thus degrade the energy efficiency. The effects of adding CP on the degradation of energy efficiency can be eliminated by enlarging $N_\mathrm{s}$ where $\beta=\frac{N_\mathrm{s}}{N}$ can approach to 1.}


In order to further simplify problem (P1), we denote $\mathbf{D}=\mathbf{F}_\mathrm{s}^\mathrm{H}\otimes \mathbf{I}_{N_\mathrm{t}} \in \mathbb{C}^{N_\mathrm{s}N_\mathrm{t}\times N_\mathrm{s}N_\mathrm{t}}$ which is a full-rank matrix 
with
\begin{align}
\mathbf{D}^{-1}=(\mathbf{F}_\mathrm{s}^\mathrm{H})^{-1}\otimes \mathbf{I}_{N_\mathrm{t}}^{-1}=\mathbf{F}_\mathrm{s}\otimes \mathbf{I}_{N_\mathrm{t}}=\mathbf{D}^\mathrm{H}.
\end{align}
Considering  $\mathbf{G}=(\mathbf{F}_\mathrm{s}^\mathrm{H}\otimes \mathbf{I}_{N_\mathrm{t}})\mathbf{X}_\mathrm{s}=\mathbf{D}\mathbf{X}_\mathrm{s}$, the optimization variable $\mathbf{X}_\mathrm{s}$ can be re-expressed as $\mathbf{X}_\mathrm{s}=\mathbf{D}^{-1}\mathbf{G}$. Let $\mathbf{H}_\mathrm{D}=\mathbf{H}_\mathrm{s}\mathbf{D}^{-1}  \in \mathbb{C}^{N_\mathrm{s}K\times N_\mathrm{s}N_\mathrm{t}}$, then problem (P1) in \eqref{eq:OFDM3} with the newly introduced power allocation constraints \eqref{PC_S} and \eqref{PC_C1} can be further  rewritten as
\begin{subeqnarray}\label{eq:OFDM04}
\hspace{-8mm}({\rm P2}): \underset{\mathbf{G}}{\min} &&\hspace{-6mm}  \frac{\rho}{\left\|\mathbf{S}\right\|_\mathrm{F}^2}
\big\| \mathbf{H}_\mathrm{D}\mathbf{G}-\mathbf{S} \big\|_\mathrm{F}^2+
\frac{1-\rho}{\left\|\mathbf{G}_{0} \right\|_\mathrm{F}^2}  \big\| \mathbf{G}-\mathbf{G}_{0} \big\|_\mathrm{F}^2 \slabel{eq:OFDM04_0} \\
\mathrm{s.t.}
&&\hspace{-6mm} \big|\left[\mathbf{G}\right]_{i,l}\big|^2 \leq \frac{\varepsilon P_\mathrm{t}}{NN_\mathrm{\mathrm{t}}}, ~\forall i\in\mathcal{N}_\mathrm{st}, ~\forall l\in\mathcal{L},  \\
&&\hspace{-6mm}  \left\| \mathbf{G} \right\|_\mathrm{F}^2=LP_\mathrm{t}^\mathrm{s}, \slabel{eq:OFDM04_2}\\
&&\hspace{-6mm}  \left\|\mathbf{\Gamma}_\mathrm{c} \mathbf{G} \right\|_\mathrm{F}^2=LP_\mathrm{t}^\mathrm{c}, \slabel{eq:OFDM04_3}
\end{subeqnarray}
which can be solved by directly optimizing the DFRC MIMO-OFDM waveform after IDFT operations, i.e.,  $\mathbf{G}\triangleq(\mathbf{F}_\mathrm{s}^\mathrm{H}\otimes \mathbf{I}_{N_\mathrm{t}})\mathbf{X}_\mathrm{s}$, and will be addressed in the next sub-section.




%

\subsection{MIMO-OFDM DFRC Waveform Design}\label{sec:Waveform_Design}
It is easy to verify that the objective function of problem (P2) in \eqref{eq:OFDM04} can be equivalently re-expressed in the form of
\begin{align}\label{Obj_AB}
\big\|\mathbf{A}\mathbf{G}-\mathbf{B}\big\|_\mathrm{F}^2
\end{align}
where
\begin{align}
&\hspace{-2mm}\mathbf{A}=\left[\frac{\sqrt{\rho}}{\left\|\mathbf{S}\right\|_\mathrm{F}}\mathbf{H}_\mathrm{D}^\mathrm{T},
\frac{\sqrt{1-\rho}}{\left\|\mathbf{G}_0\right\|_\mathrm{F}}\mathbf{I}_{N_\mathrm{s}N_\mathrm{t}}\right]^\mathrm{T} \hspace{-2mm} \in \mathbb{C}^{(N_\mathrm{s}K+N_\mathrm{s}N_\mathrm{t})\times N_\mathrm{s}N_\mathrm{t}},  \\
&\hspace{-2mm}\mathbf{B}=\left[ \frac{\sqrt{\rho}}{\left\|\mathbf{S}\right\|_\mathrm{F}}\mathbf{S}^\mathrm{T},
\frac{\sqrt{1-\rho}}{\left\|\mathbf{G}_0\right\|_\mathrm{F}}\mathbf{G}_0^\mathrm{T} \right]^\mathrm{T} \hspace{-2mm}\in \mathbb{C}^{(N_\mathrm{s}K+N_\mathrm{s}N_\mathrm{t})\times L}.
\end{align}
To further facilitate the optimization, we then transform problem (P2) in \eqref{eq:OFDM04} with the objective function \eqref{Obj_AB} into an equivalent vector form as below
\begin{subeqnarray}\label{eq:OFDM042}
({\rm P3}): \underset{\mathbf{g}}{\min} &&\hspace{-4mm} \Big\| \widetilde{\mathbf{A}}\mathbf{g}-\mathbf{b} \Big\|^2
\slabel{eq:OFDM042_0} \\
\mathrm{s.t.}
&&\hspace{-4mm} |g_{q}|^2 \leq \frac{\varepsilon P_\mathrm{t}}{NN_\mathrm{\mathrm{t}}}, ~\forall q \in \mathcal{N}_{\mathrm{st}L},  \slabel{eq:OFDM042_1}\\
&&\hspace{-4mm} \left\| \mathbf{g} \right\|^2=LP_\mathrm{t}^\mathrm{s}, \slabel{eq:OFDM042_2} \\
&&\hspace{-4mm}  \Big\|\widetilde{\mathbf{\Gamma}}_\mathrm{c} \mathbf{g} \Big\|^2=LP_\mathrm{t}^\mathrm{c}, \slabel{eq:OFDM042_3}
\end{subeqnarray}
where we have
\begin{align}\label{ABg}
&\widetilde{\mathbf{A}}=\mathbf{I}_L\otimes\mathbf{A} \in \mathbb{C}^{(N_\mathrm{s}K+N_\mathrm{s}N_\mathrm{t})L\times N_\mathrm{s}N_\mathrm{t}L},  \\
&\mathbf{g}=\mathrm{vec}(\mathbf{G}) \in \mathbb{C}^{N_\mathrm{s}N_\mathrm{t}L\times 1}, \\
&\mathbf{b}=\mathrm{vec}(\mathbf{B}) \in \mathbb{C}^{(N_\mathrm{s}K+N_\mathrm{s}N_\mathrm{t})L\times 1},\\
&\widetilde{\mathbf{\Gamma}}_\mathrm{c}=\mathbf{I}_L\otimes\mathbf{\Gamma}_\mathrm{c} \in \mathbb{C}^{ N_\mathrm{s}N_\mathrm{t}L \times N_\mathrm{s}N_\mathrm{t}L },
\end{align}
and $g_{q}$ indicates the $q$-th element of the vector $\mathbf{g}$ for $q \in \mathcal{N}_{\mathrm{st}L}=\{1,\cdots, N_\mathrm{s}N_\mathrm{t}L\}$.

Moreover, the minimization of the objective function in \eqref{eq:OFDM042_0} without any constraints can be equivalently transformed into a problem with an auxiliary parameter $\xi$ as follows
\begin{align}\label{Obj_extend}
\underset{\mathbf{g}}{\min} ~\Big\| \widetilde{\mathbf{A}}\mathbf{g}-\xi\mathbf{b} \Big\|^2, ~~\mathrm{s.t.}~|\xi|^2=1.
\end{align}
Note that if $(\mathbf{g}^\mathrm{o},\xi^\mathrm{o})$ is the optimal solution of problem \eqref{Obj_extend}, then $\mathbf{g}^\mathrm{o}(\xi^\mathrm{o})^*$ is the optimal solution for minimizing the objective function in \eqref{eq:OFDM042_0}. Based on this observation, problem (P3) in \eqref{eq:OFDM042} can be equivalently reformulated as
\begin{subeqnarray}\label{eq:OFDM043}
({\rm P4}): \underset{\mathbf{g},\xi}{\min} &&\hspace{-4mm} [\mathbf{g}^\mathrm{H}, \xi^*]
\left[
 \begin{matrix}
   \widetilde{\mathbf{A}}^\mathrm{H}\widetilde{\mathbf{A}}   &-\widetilde{\mathbf{A}}^\mathrm{H}\mathbf{b} \\
   -\mathbf{b}^\mathrm{H}\widetilde{\mathbf{A}}   &\mathbf{b}^\mathrm{H}\mathbf{b}
  \end{matrix}
\right]
\left[
 \begin{matrix}
   \mathbf{g}  \\
   \xi
  \end{matrix}
\right]
\slabel{eq:OFDM043_0} \\
\mathrm{s.t.}
&&\hspace{-4mm} \mathrm{diag}(\mathbf{gg}^\mathrm{H}) \leq \frac{\varepsilon P_\mathrm{t}}{NN_\mathrm{\mathrm{t}}}\mathbf{1}_{N_\mathrm{s}N_\mathrm{t}L},  \slabel{eq:OFDM043_1}\\
&&\hspace{-4mm} \left\| \mathbf{g} \right\|^2=LP_\mathrm{t}^\mathrm{s}, \slabel{eq:OFDM043_3} \\
&&\hspace{-4mm} \Big\|\widetilde{\mathbf{\Gamma}}_\mathrm{c} \mathbf{g} \Big\|^2=LP_\mathrm{t}^\mathrm{c}, \slabel{eq:OFDM043_4}\\
&&\hspace{-4mm} |\xi|^2=1, \slabel{eq:OFDM043_2}
\end{subeqnarray}
which is a homogeneous quadratically constrained quadratic program (QCQP) and can be solved by SDR. We denote
\begin{align}\label{QCQP_1}
&\hspace{-2mm}\widehat{\mathbf{g}}=[\mathbf{g}^\mathrm{H}, \xi^*]^\mathrm{H}\in \mathbb{C}^{(N_\mathrm{s}N_\mathrm{t}L+1)\times 1},  \\
&\hspace{-2mm}\widehat{\mathbf{G}}=\widehat{\mathbf{g}}\widehat{\mathbf{g}}^\mathrm{H}\in \mathbb{C}^{(N_\mathrm{s}N_\mathrm{t}L+1)\times (N_\mathrm{s}N_\mathrm{t}L+1)},  \\
&\hspace{-2mm}\mathbf{Q}=
\left[
 \begin{matrix}
   \widetilde{\mathbf{A}}^\mathrm{H}\widetilde{\mathbf{A}}   &-\widetilde{\mathbf{A}}^\mathrm{H}\mathbf{b} \\
   -\mathbf{b}^\mathrm{H}\widetilde{\mathbf{A}}              &\mathbf{b}^\mathrm{H}\mathbf{b}
  \end{matrix}
\right]\in \mathbb{C}^{(N_\mathrm{s}N_\mathrm{t}L+1)\times (N_\mathrm{s}N_\mathrm{t}L+1)},  \label{Matrix_Q}\\
&\hspace{-2mm}\widehat{\mathbf{\Gamma}}_\mathrm{c}=
\left[
 \begin{matrix}
   \widetilde{\mathbf{\Gamma}}_\mathrm{c}^\mathrm{H}\widetilde{\mathbf{\Gamma}}_\mathrm{c}   &\hspace{-4mm}\mathbf{0}_{ N_\mathrm{s}N_\mathrm{t}L \times 1}\\
    \mathbf{0}_{1 \times  N_\mathrm{s}N_\mathrm{t}L }                                        &\hspace{-4mm}0
  \end{matrix}
\right]\hspace{-1mm}\in \mathbb{C}^{ (N_\mathrm{s}N_\mathrm{t}L+1) \times (N_\mathrm{s}N_\mathrm{t}L+1)},
\end{align}
problem (P4) in \eqref{eq:OFDM043} can be further rewritten as
\begin{subeqnarray}\label{eq:OFDM005}
\hspace{-4mm}({\rm P5}): \underset{\widehat{\mathbf{G}}}{\min} &&\hspace{-4mm} \mathrm{Tr}(\mathbf{Q}\widehat{\mathbf{G}})  \slabel{eq:OFDM005_0} \\
\mathrm{s.t.}
&&\hspace{-4mm} \widehat{\mathbf{G}}_{q,q}\leq \frac{\varepsilon P_\mathrm{t}}{NN_\mathrm{\mathrm{t}}}, ~ \forall q \in \mathcal{N}_\mathrm{stL},   \slabel{eq:OFDM005_1}\\
&&\hspace{-4mm} \mathrm{Tr}( \widehat{\mathbf{G}} )=LP_\mathrm{t}^\mathrm{s}+1, \slabel{eq:OFDM005_3} \\
&&\hspace{-4mm} \mathrm{Tr}( \widehat{\mathbf{\Gamma}}_\mathrm{c}\widehat{\mathbf{G}} )=LP_\mathrm{t}^\mathrm{c}, \slabel{eq:OFDM005_6} \\
&&\hspace{-4mm} \widehat{\mathbf{G}}_{N_\mathrm{s}N_\mathrm{t}L+1,N_\mathrm{s}N_\mathrm{t}L+1}=1,   \slabel{eq:OFDM005_2} \\
&&\hspace{-4mm} \widehat{\mathbf{G}}\succeq \mathbf{0},  \slabel{eq:OFDM005_7} \\ 
&&\hspace{-4mm} \mathrm{rank}(\widehat{\mathbf{G}})=1, \slabel{eq:OFDM005_4}
\end{subeqnarray}
through optimizing the hermitian semi-definite matrix variable  $\widehat{\mathbf{G}}$.  
By dropping the rank-1 constraint on $\widehat{\mathbf{G}}$ in \eqref{eq:OFDM005_4}, problem (P5) in \eqref{eq:OFDM005} becomes a standard SDP and can be effectively solved  by the classic SDR technique via the existing numerical tools, e.g., CVX. 
\begin{lemma}\label{lemma1}
It can be proved that the rank-1 solution exist in general for the SDR of problem (P5) in \eqref{eq:OFDM005}, indicating that the globally optimal solution of problem (P5) can always be obtained by solving its SDR without considering the rank-1 constraint \eqref{eq:OFDM005_4}, denoted as $\widehat{\mathbf{G}}^\mathrm{o}$. 
\end{lemma}

\begin{proof}
See Appendix~\ref{sec:proof_lem1}.
\end{proof}


Once we obtain the optimal solution of $\mathbf{G}^\mathrm{o}$ based on $\widehat{\mathbf{G}}^\mathrm{o}$, then the optimal precoded DFRC MIMO-OFDM waveform on the symbol subcarriers can be expressed as $\mathbf{X}_\mathrm{s}^\mathrm{o}=\mathbf{D}_\mathrm{s}^{-1}\mathbf{G}^\mathrm{o}$.

The previously proposed SDR method for solving problem (P5) in \eqref{eq:OFDM005} can obtain the global optimal solution of the original problem (P2) in \eqref{eq:OFDM04}, but its computation complexity is quite high. 
We can observe that the number of optimization variables for problem (P5)  is $(N_\mathrm{s}N_\mathrm{t}L+1)^2$ which is nearly proportional to the square of $N_\mathrm{s}N_\mathrm{t}L$, and thus high computational complexity may hinder its use in practical systems, especially considering the wideband scenario with longer coherent frame, i.e., with larger $L$. Next, we propose a method with much lower complexity by dividing the original  problem (P2)  into $L$ subproblems corresponding to each  frame/snapshot of the DFRC MIMO-OFDM waveform.

\vspace{-4mm}
\subsection{Low-Complexity Approach}\label{sec:Waveform_Design_low}
In order to solve problem (P2) in \eqref{eq:OFDM04} with low complexity, we further transform problem (P2)  into the following form
\begin{subeqnarray}\label{eq:OFDM07}
\hspace{-4mm}({\rm P6}): \underset{  \{\mathbf{g}_l\}_{l\in \mathcal{L}} }{\min} &&\hspace{-4mm} \sum_{l=1}^L \bigg(\frac{\rho}{\left\|\mathbf{S}\right\|_\mathrm{F}^2}
\big\| \mathbf{H}_{\mathrm{D}}\mathbf{g}_l-\mathbf{s}_l \big\|^2+ \nonumber \\
&&~~~~~~~\hspace{-4mm}\frac{1-\rho}{\left\|\mathbf{G}_{0}\right\|_\mathrm{F}^2} \big\| \mathbf{g}_l-\mathbf{g}_{0,l} \big\|^2 \bigg) \slabel{eq:OFDM07_0} \\
\mathrm{s.t.} ~
&&\hspace{-4mm} |g_{l,i}|^2\leq \frac{\varepsilon P_\mathrm{t}}{NN_\mathrm{\mathrm{t}}}, ~\forall i\in\mathcal{N}_\mathrm{st}, ~\forall l\in\mathcal{L}, \slabel{eq:OFDM07_1}\\
&&\hspace{-4mm} \left\| \mathbf{g}_l \right\|^2=P_\mathrm{t}^\mathrm{s}, ~\forall l\in\mathcal{L},  \slabel{eq:OFDM07_2} \\
&&\hspace{-4mm}  \left\|\mathbf{\Gamma}_\mathrm{c} \mathbf{g}_l \right\|^2=P_\mathrm{t}^\mathrm{c}, ~\forall l\in\mathcal{L}, \slabel{eq:OFDM07_3}
\end{subeqnarray}
where $\mathbf{g}_l \in \mathbb{C}^{N_\mathrm{s}N_\mathrm{t} \times 1}$, $\mathbf{s}_l \in \mathbb{C}^{N_\mathrm{s}K \times 1}$ and $\mathbf{g}_{0,l} \in \mathbb{C}^{N_\mathrm{s}N_\mathrm{t} \times 1}$ correspond to the $l$-th column of $\mathbf{G}$, $\mathbf{S}$ and $\mathbf{G}_0$, respectively. Also,  $g_{l,i}$ indicates the $i$-th element of $\mathbf{g}_l$.
In problem (P6),  the original sum power allocation constraints in \eqref{eq:OFDM04_2} and \eqref{eq:OFDM04_3}  of problem (P2) is respectively further relaxed into $L$ individual power allocation constraints  as in \eqref{eq:OFDM07_2} and \eqref{eq:OFDM07_3}  for $l \in \mathcal{L}$, which is optimal in the  average basis and easy to operate in practical communication systems.
Note that problem (P6) can be parallelly solved by addressing $L$ sub-problems as shown in the following $l$-th sub-problem (P6.1)
\begin{subeqnarray}\label{eq:OFDM08}
\hspace{-7mm}({\rm P6.1}): \underset{  \mathbf{g}_l }{\min} &&\hspace{-7mm}  \frac{\rho}{ \left\| \mathbf{S}\right\|_\mathrm{F}^2 }
 \left\| \mathbf{H}_{\mathrm{D}}\mathbf{g}_l-\mathbf{s}_l \right\|^2+ \frac{1-\rho}{ \left\|\mathbf{G}_{0}\right\|_\mathrm{F}^2} \left\| \mathbf{g}_l-\mathbf{g}_{0,l} \right\|^2  \slabel{eq:OFDM08_0} \\
\mathrm{s.t.} ~
&&\hspace{-7mm} |g_{l,i}|^2 \leq \frac{\varepsilon P_\mathrm{t}}{NN_\mathrm{\mathrm{t}}},  ~\forall i\in\mathcal{N}_\mathrm{st},  \slabel{eq:OFDM08_1}\\
&&\hspace{-7mm} \left\|\mathbf{g}_l \right\|^2=P_\mathrm{t}^\mathrm{s},   \slabel{eq:OFDM08_2}\\
&&\hspace{-7mm} \left\|\mathbf{\Gamma}_\mathrm{c} \mathbf{g}_l \right\|^2=P_\mathrm{t}^\mathrm{c},  \slabel{eq:OFDM08_3}
\end{subeqnarray}
%
%
which can be finally transformed into the following form
\begin{subeqnarray}\label{eq:OFDM008}
\hspace{-6mm}({\rm P6.2}): \underset{\widehat{\mathbf{G}}_l}{\max} &&\hspace{-4mm} \mathrm{Tr}(\mathbf{Q}_l\widehat{\mathbf{G}}_l)  \slabel{eq:OFDM008_0} \\
\mathrm{s.t.} ~
&&\hspace{-4mm} \widehat{\mathbf{G}}_{l,i,i}\leq \frac{\varepsilon P_\mathrm{t}}{NN_\mathrm{\mathrm{t}}},  ~\forall i\in\mathcal{N}_\mathrm{st},    \slabel{eq:OFDM008_1}\\
&&\hspace{-4mm} \mathrm{Tr}( \widehat{\mathbf{G}}_l )=P_\mathrm{t}^\mathrm{s}+1, \slabel{eq:OFDM008_3} \\
&&\hspace{-4mm} \mathrm{Tr}( \dot{\mathbf{\Gamma}}_\mathrm{c} \widehat{\mathbf{G}}_l )=P_\mathrm{t}^\mathrm{c}, \slabel{eq:OFDM008_6} \\
&&\hspace{-4mm} \widehat{\mathbf{G}}_{l,N_\mathrm{s}N_\mathrm{t}+1,N_\mathrm{s}N_\mathrm{t}+1}=1,   \slabel{eq:OFDM008_2}\\
&&\hspace{-4mm} \widehat{\mathbf{G}}_l\succeq 0,  \slabel{eq:OFDM008_4}\\
&&\hspace{-4mm} \mathrm{rank}(\widehat{\mathbf{G}}_l)=1, \slabel{eq:OFDM008_5}
\end{subeqnarray}
which is similar to problem (P5) in \eqref{eq:OFDM005} where
\begin{align}\label{QCQP_2}
&\widehat{\mathbf{G}}_l=\widehat{\mathbf{g}}_l\widehat{\mathbf{g}}_l^\mathrm{H}\in \mathbb{C}^{(N_\mathrm{s}N_\mathrm{t}+1)\times (N_\mathrm{s}N_\mathrm{t}+1)},  \\
&\mathbf{Q}_l=
\left[
 \begin{matrix}
   \mathbf{A}^\mathrm{H}\mathbf{A}   &-\mathbf{A}^\mathrm{H}\mathbf{b}_l \\
   -\mathbf{b}_l^\mathrm{H}\mathbf{A}              &\mathbf{b}_l^\mathrm{H}\mathbf{b}_l
  \end{matrix}
\right]\in \mathbb{C}^{(N_\mathrm{s}N_\mathrm{t}+1)\times (N_\mathrm{s}N_\mathrm{t}+1)}  \\
&\dot{\mathbf{\Gamma}}_\mathrm{c}=
\left[
 \begin{matrix}
   \mathbf{\Gamma}_\mathrm{c}^\mathrm{H}\mathbf{\Gamma}_\mathrm{c}   &\hspace{-4mm}\mathbf{0}_{ N_\mathrm{s}N_\mathrm{t} \times 1}\\
    \mathbf{0}_{1 \times  N_\mathrm{s}N_\mathrm{t} }                                        &\hspace{-4mm}0
  \end{matrix}
\right]\hspace{-1mm}\in \mathbb{C}^{ (N_\mathrm{s}N_\mathrm{t}+1) \times (N_\mathrm{s}N_\mathrm{t}+1)},
\end{align}
with $\widehat{\mathbf{g}}_l=[\mathbf{g}_l^\mathrm{H}, \xi_l^*]\in \mathbb{C}^{(N_\mathrm{s}N_\mathrm{t}+1)\times 1}$ and $\mathbf{b}_l \in \mathbb{C}^{(N_\mathrm{s}K+N_\mathrm{s}N_\mathrm{t})\times 1}$ being the $l$-th column of $\mathbf{B}$.
The SDR of problem (P6.2) in \eqref{eq:OFDM008} can also be solved with CVX by dropping the rank-1 constraint \eqref{eq:OFDM008_5}, and it can be verified that the rank-1 solution exist in general via the same method of Lemma \ref{lemma1}.

The number of optimization variables for solving problem (P6.2) in \eqref{eq:OFDM008} is $(N_\mathrm{s}N_\mathrm{t}+1)^2$, and thus the computational complexity/running time can be significantly reduced comparing with solving problem (P5) in \eqref{eq:OFDM005}.
In the simulation results, we leverage the low-complexity method of solving sub-problems proposed in this section due to the fact that it is easy to exceed the array size of MATLAB operation when solving the original problem (P5) in \eqref{eq:OFDM005} with the SDR algorithm given in Section \ref{sec:Waveform_Design}, especially with larger values of $N_\mathrm{s}$, $N_\mathrm{t}$ and $L$.

\subsection{Desired Benchmark of Radar Waveform with Directional Beampattern Design}\label{sec:Waveform_DBD}
In this subsection, we provide a desired benchmark for the radar detection waveform of the MIMO-OFDM system by utilizing the technique of directional beampattern design, denoted as $\mathbf{G}_0^\mathrm{d}=(\mathbf{F}_\mathrm{s}^\mathrm{H}\otimes \mathbf{I}_{N_\mathrm{t}})\mathbf{X}_{\mathrm{s},0}^\mathrm{d}$, where $\mathbf{X}_{\mathrm{s},0}^\mathrm{d}$ is the corresponding precoded waveform.
Assuming  $\mathbf{R}_\mathrm{d}  \in \mathbb{C}^{N_\mathrm{t}\times N_\mathrm{t}}$ is a Hermitian positive semi-definite covariance matrix corresponding to a well designed MIMO radar beampattern on a single carrier,  the radar detection waveform of the considered MIMO-OFDM system via the directional beampattern design can be obtained  by solving the following MUI minimization problem
\begin{subeqnarray}\label{eq:DBD_OFDM}
 \underset{\mathbf{G}}{\min} &&\hspace{-4mm} \left\| \mathbf{H}_\mathrm{D}\mathbf{G}-\mathbf{S} \right\|_\mathrm{F}^2  \slabel{eq:DBD_OFDM_0} \\
\mathrm{s.t.}
&&\hspace{-4mm}  \frac{1}{L}\mathbf{GG}^\mathrm{H}= \mathbf{I}_{N_\mathrm{s}}\otimes \mathbf{R}_\mathrm{d}. \slabel{eq:DBD_OFDM_2}
\end{subeqnarray}

Considering the Cholesky decomposition on $\mathbf{R}_\mathrm{d}$, we have
\begin{align}
\mathbf{I}_{N_\mathrm{s}}\otimes\mathbf{R}_\mathrm{d}=\mathbf{\Phi\Phi}^\mathrm{H},
\end{align}
where $\mathbf{\Phi} \in \mathbb{C}^{N_\mathrm{s}N_\mathrm{t}\times N_\mathrm{s}N_\mathrm{t}}$ is a lower triangular matrix. Without loss of generality, we assume $\mathbf{R}_\mathrm{d}$ is positive-definite such that  $\mathbf{I}_{N_\mathrm{s}}\otimes \mathbf{R}_\mathrm{d}$ is positive-definite as well  which can guarantee that $\mathbf{\Phi}$ is invertible. Hence, the constraint in \eqref{eq:DBD_OFDM_2} can be equivalently re-expressed as
\begin{align}
\frac{1}{L}\mathbf{\Phi}^{-1}\mathbf{G}(\mathbf{\Phi}^{-1}\mathbf{G})^\mathrm{H}= \mathbf{I}_{N_\mathrm{s}N_\mathrm{t}}.
\end{align}
Let us denote $\mathbf{G}_\Phi=\sqrt{\frac{1}{L}}\mathbf{\Phi}^{-1}\mathbf{G}$, then problem \eqref{eq:DBD_OFDM} can be reformulated as
\begin{subeqnarray}\label{eq:DBD_OFDM2}
 \underset{\mathbf{G}_\Phi}{\min} &&\hspace{-4mm}  \left\|\sqrt{L}\mathbf{H}_\mathrm{D}\mathbf{\Phi}\mathbf{G}_\Phi-\mathbf{S} \right\|_\mathrm{F}^2  \slabel{eq:DBD_OFDM2_0} \\
\mathrm{s.t.}
&&\hspace{-4mm}  \mathbf{G}_\Phi \mathbf{G}_\Phi^\mathrm{H}= \mathbf{I}_{N_\mathrm{s}N_\mathrm{t}}, \slabel{eq:DBD_OFDM2_2}
\end{subeqnarray}
which has been proven as the Orthogonal Procrustes problem (OPP) \cite{PHD2008_T.Viklands2006Algorithms}, and a simple closed-form global optimal solution based on the Singular Value Decomposition (SVD) is given as
\begin{align}
\mathbf{G}_\Phi=\mathbf{U}\mathbf{I}_{N_\mathrm{s}N_\mathrm{t} \times L}\mathbf{V}^\mathrm{H},
\end{align}
where $\mathbf{U}\mathbf{\Sigma}\mathbf{V}^\mathrm{H}= \mathbf{\Phi}^\mathrm{H}\mathbf{H}_\mathrm{D}^\mathrm{H}\mathbf{S}$ is the SVD of $\mathbf{\Phi}^\mathrm{H}\mathbf{H}_\mathrm{D}^\mathrm{H}\mathbf{S}$. Hence, the optimal solution of the original problem \eqref{eq:DBD_OFDM} with directional radar beampattern design can be expressed as
\begin{align}
\mathbf{G}_0^\mathrm{d}=\sqrt{L}\mathbf{\Phi}\mathbf{G}_\Phi=\sqrt{L}\mathbf{\Phi}\mathbf{U}\mathbf{I}_{N_\mathrm{s}N_\mathrm{t} \times L}\mathbf{V}^\mathrm{H},
\end{align}
which can be used as a desired benchmark of radar waveform, i.e., $\mathbf{G}_0$, used in  problem (P2) in \eqref{eq:OFDM04}.


\section{Low-PAPR DFRC MIMO-OFDM Waveform Design with Oversampling}\label{sec:Problem_CPOS}
OFDM introduces large amplitude variations in time, which may result in significant signal distortion in the presence of non-linear amplifiers. In practical OFDM systems utilizing low-cost non-linear amplifiers, the technique of oversampling is usually required for digital pre-distortion to avoid serious distortion of the time-domain signals \cite{2003TCOM_M.Sharif_OnthePAPR}. Meanwhile, we can obtain a more accurate PAPR measurement through the oversampled  signals compared with the scenario leveraging Nyquist-rate sampling in Section \ref{sec:Problem_CP}  (oversampling rate equals to 1). 
Based on the results in \cite{B2010_Y.Cho_MIMOOFDM,WC2005_S.Han_AnOverview},  PAPR levels can be accurately measured if the discrete-time signals are $\Upsilon$-times interpolated (oversampled) with $\Upsilon\geq 4$.
Hence, in this section, we further address the low-PAPR DFRC MIMO-OFDM waveform design in the practical  scenario  with oversampling.

\subsection{Problem Formulation}\label{sec:PF_CP_OS1}
Considering an oversampling rate $\Upsilon\geq 4$, the corresponding IDFT of each transmit antenna should with  $\Upsilon N_\mathrm{s}$  input and output points. We further denote $\mathbf{F}_\mathrm{os} \in \mathbb{C}^{\Upsilon N_\mathrm{s}\times \Upsilon N_\mathrm{s}}$ as the normalized oversampling DFT matrix where $F_{i,m}^\mathrm{os}=\frac{1}{\sqrt{ N_\mathrm{s}}}e^{-j\frac{2\pi}{\Upsilon N_\mathrm{s}}(i-1)(m-1)}$ is the ($i,m$)-th element with index $i,m \in \mathcal{N}_\mathrm{os}=\{1,2,\cdots, \Upsilon N_\mathrm{s}\}$. In addition, the precoded symbol matrix should be $\Upsilon$-times interpolated with 0 as
\begin{align}\label{Xos}
\mathbf{X}_\mathrm{os}=\Big[&\mathbf{X}_1^\mathrm{T},\cdots, \mathbf{X}_{N_\mathrm{s}/2}^\mathrm{T},~
\underbrace{\mathbf{0}, ~~\cdots, ~~\mathbf{0}}_{\Upsilon N_\mathrm{s}-1 ~ \mathbf{0}_{L\times N_\mathrm{t}}  },~ \nonumber\\
&~~~~~~~~~\mathbf{X}_{N_\mathrm{s}/2+1}^\mathrm{T}, \cdots, \mathbf{X}_{N_\mathrm{s}}^\mathrm{T}\Big]^\mathrm{T} \in \mathbb{C}^{\Upsilon N_\mathrm{s}N_\mathrm{t}\times L},
\end{align}
considering even number of subcarriers $N_\mathrm{s}$, then the oversampled IDFT output can be written as $(\mathbf{F}_\mathrm{os}^\mathrm{H}\otimes \mathbf{I}_{N_\mathrm{t}})\mathbf{X}_\mathrm{os} \in \mathbb{C}^{\Upsilon N_\mathrm{s}N_\mathrm{t}\times L}$. Hence, the PAPR constraint of the oversampled discrete-time signal for the DFRC MIMO-OFDM waveform  after IDFT can be expressed as
\begin{align}\label{PAPR2}
\mathrm{PAPR}(\mathbf{X}_\mathrm{os})&=\frac{\underset{\gamma,l}{\max} ~\left|\big[(\mathbf{F}_\mathrm{os}^\mathrm{H}\otimes \mathbf{I}_{N_\mathrm{t}})\mathbf{X}_\mathrm{os}\big]_{\gamma,l}\right|^2}
{\frac{1}{\Upsilon NN_\mathrm{t}L} \left\| ( \dot{\mathbf{F}}_\mathrm{os}^\mathrm{H}\otimes \mathbf{I}_{N_\mathrm{t}})\dot{\mathbf{X}}_\mathrm{os} \right\|_\mathrm{F}^2} \nonumber\\
&\overset{(a)}{=}\frac{\underset{\gamma,l}{\max} ~\left|\big[(\mathbf{F}_\mathrm{os}^\mathrm{H}\otimes \mathbf{I}_{N_\mathrm{t}})\mathbf{X}_\mathrm{os}\big]_{\gamma,l}\right|^2}
{\frac{1}{ NN_\mathrm{t}L}\left\|( \dot{\mathbf{F}}^\mathrm{H}\otimes \mathbf{I}_{N_\mathrm{t}})\dot{\mathbf{X}} \right\|_\mathrm{F}^2} \nonumber\\
&\overset{(b)}{=}\frac{\underset{\gamma,l}{\max} ~\left|\big[(\mathbf{F}_\mathrm{os}^\mathrm{H}\otimes \mathbf{I}_{N_\mathrm{t}})\mathbf{X}_\mathrm{os}\big]_{\gamma,l}\right|^2 }
{\frac{1}{ NN_\mathrm{t}} P_\mathrm{t} }\leq\varepsilon, 
\end{align}
where $\gamma \in N_\mathrm{\gamma st}=\{1,\cdots, \Upsilon N_\mathrm{s}N_\mathrm{t}\}$, $l\in\mathcal{L}$ and
\begin{align}
&\dot{\mathbf{F}}_\mathrm{os}=[\mathbf{F}_\mathrm{os}^\mathrm{c}, \mathbf{F}_\mathrm{os}]\in \mathbb{C}^{\Upsilon N_\mathrm{s}\times (\Upsilon N_\mathrm{s}+N_\mathrm{c}) },\\
&\dot{\mathbf{X}}_\mathrm{os}=[(\mathbf{X}_\mathrm{os}^\mathrm{c})^\mathrm{T}, \mathbf{X}_\mathrm{os}^\mathrm{T}]\in \mathbb{C}^{ (\Upsilon N_\mathrm{s}+N_\mathrm{c})N_\mathrm{t} \times L } \label{Xos_CP}
\end{align}
with $\mathbf{F}_\mathrm{os}^\mathrm{c} \in \mathbb{C}^{\Upsilon N_\mathrm{s}\times N_\mathrm{c}}$ being formulated by the last $N_\mathrm{c}$ columns of $\mathbf{F}_\mathrm{os}$ while $\mathbf{X}_\mathrm{os}^\mathrm{c}  \in \mathbb{C}^{ N_\mathrm{c}\times N_\mathrm{t}}$ being formulated by the last $N_\mathrm{c}N_\mathrm{t}$ rows of $\mathbf{X}_\mathrm{os}$, similar to the definitions of $\mathbf{F}_\mathrm{c}$ in \eqref{Fc} and $\mathbf{X}_\mathrm{c}$ in \eqref{Xc}, respectively. It is easy to observe that $\mathbf{X}_\mathrm{os}^\mathrm{c}=\mathbf{X}_\mathrm{c}$ in the practical scenario where $N_\mathrm{c}\leq N_\mathrm{s}/2$, which is considered in this paper.\footnote{In order to improve the energy efficiency of practical OFDM communication systems,  the number of symbols is usually set to be much larger than  the length of CP, i.e., $N_\mathrm{s}\gg N_\mathrm{c}$. In the scenario considered in this paper, we assume that $N_\mathrm{s}/2\geq N_\mathrm{c}$.} 
In \eqref{PAPR2}, (a) and (b) are based on the fact that
\begin{align}
\frac{1}{\Upsilon }\big\| ( \dot{\mathbf{F}}_\mathrm{os}^\mathrm{H}\otimes \mathbf{I}_{N_\mathrm{t}})\dot{\mathbf{X}}_\mathrm{os} \big\|_\mathrm{F}^2=
\big\|( \dot{\mathbf{F}}^\mathrm{H}\otimes \mathbf{I}_{N_\mathrm{t}})\dot{\mathbf{X}} \big\|_\mathrm{F}^2=LP_\mathrm{t}.
\end{align}
Hence, the PAPR constraint in \eqref{PAPR2} is  equivalent to the set of PAPR constraints given below:
\begin{align}\label{PAPR_acc1}
\Big|\big[(\mathbf{F}_\mathrm{os}^\mathrm{H}\otimes \mathbf{I}_{N_\mathrm{t}})\mathbf{X}_\mathrm{os}\big]_{\gamma,l}\Big|^2 \leq \frac{\varepsilon P_\mathrm{t}}{ NN_\mathrm{\mathrm{t}}}, ~\forall \gamma \in N_\mathrm{\gamma st}, ~\forall l\in\mathcal{L}.
\end{align}

It is easy to note the effective elements of the oversampled matrix $\mathbf{X}_\mathrm{os} \in \mathbb{C}^{\Upsilon N_\mathrm{s}N_\mathrm{t}\times L}$  in \eqref{Xos} are exactly the elements of the original precoded symbol matrix $\mathbf{X}_\mathrm{s} \in \mathbb{C}^{N_\mathrm{s}N_\mathrm{t}\times L}$. For utilizing the accurate PAPR constraint \eqref{PAPR_acc1} and facilitating the formulation of the optimization problem for the effective matrix $\mathbf{X}_\mathrm{s}$, one challenging lies in transforming $(\mathbf{F}_\mathrm{os}^\mathrm{H}\otimes \mathbf{I}_{N_\mathrm{t}})\mathbf{X}_\mathrm{os}$ in \eqref{PAPR_acc1} into a function of $\mathbf{X}_\mathrm{s}$, which is also a key step for simplifying the problem solving process and reducing the computational complexity.
To this end, we further denote an equivalent oversampling DFT matrix $\widetilde{\mathbf{F}}_\mathrm{os} \in \mathbb{C}^{\Upsilon N_\mathrm{s}\times N_\mathrm{s}}$ as
\begin{align}\label{DFT_E1}
\hspace{-2mm}\widetilde{\mathbf{F}}_\mathrm{os}&=\left\{
\begin{aligned}
&\frac{1}{\sqrt{ N_\mathrm{s}}}e^{-j\frac{2\pi}{\Upsilon N_\mathrm{s}}(n-1)(m-1)}, \\
&~~~~~~n=1,\cdots, \frac{N_\mathrm{s}}{2}, ~m=1, \cdots, \Upsilon N_\mathrm{s}, \\
&\frac{1}{\sqrt{ N_\mathrm{s}}}e^{-j\frac{2\pi}{\Upsilon N_\mathrm{s}}(\Upsilon N_\mathrm{s}-N_\mathrm{s}+(n-1) )(m-1)}, \\
&~~~~~~ n=\frac{N_\mathrm{s}}{2}+1\cdots, N_\mathrm{s}, ~m=1, \cdots, \Upsilon N_\mathrm{s}, \\
\end{aligned}\right.
\end{align}
and it is easy to prove that
\begin{align}
&(\mathbf{F}_\mathrm{os}^\mathrm{H}\otimes \mathbf{I}_{N_\mathrm{t}})\mathbf{X}_\mathrm{os}=(\widetilde{\mathbf{F}}_\mathrm{os}^\mathrm{H}\otimes \mathbf{I}_{N_\mathrm{t}})\mathbf{X}_\mathrm{s} \in \mathbb{C}^{\Upsilon N_\mathrm{s}N_\mathrm{t}\times L}.
\end{align}
Then the PAPR constraints in \eqref{PAPR_acc1} can be equivalently transformed into
\begin{align}\label{PAPR_acc2}
\Big|\big[(\widetilde{\mathbf{F}}_\mathrm{os}^\mathrm{H}\otimes \mathbf{I}_{N_\mathrm{t}})\mathbf{X}_\mathrm{s}\big]_{\gamma,l}\Big|^2 \leq \frac{\varepsilon P_\mathrm{t}}{ NN_\mathrm{\mathrm{t}}}, ~\forall \gamma \in N_\mathrm{\gamma st}, ~\forall l\in\mathcal{L}.
\end{align}
Considering the equivalent IDFT operation $\widetilde{\mathbf{F}}_\mathrm{os}^\mathrm{H}\otimes \mathbf{I}_{N_\mathrm{t}} \in \mathbb{C}^{\Upsilon N_\mathrm{s}N_\mathrm{t}\times N_\mathrm{s}N_\mathrm{t} }$ and DFT operation $\widetilde{\mathbf{F}}_\mathrm{os}\otimes \mathbf{I}_{N_\mathrm{t}} \in \mathbb{C}^{N_\mathrm{s}N_\mathrm{t}\times\Upsilon N_\mathrm{s}N_\mathrm{t} }$ over the MIMO-OFDM system, we further have
\begin{align}
&(\widetilde{\mathbf{F}}_\mathrm{os}\otimes \mathbf{I}_{N_\mathrm{t}})(\widetilde{\mathbf{F}}_\mathrm{os}^\mathrm{H}\otimes \mathbf{I}_{N_\mathrm{t}})= \mathbf{I}_{N_\mathrm{s}N_\mathrm{t}} \Leftrightarrow \nonumber\\
&(\widetilde{\mathbf{F}}_\mathrm{os}^\mathrm{H}\otimes \mathbf{I}_{N_\mathrm{t}})^{\dag}=\widetilde{\mathbf{F}}_\mathrm{os}\otimes \mathbf{I}_{N_\mathrm{t}}=(\widetilde{\mathbf{F}}_\mathrm{os}^\mathrm{H}\otimes \mathbf{I}_{N_\mathrm{T}})^{\mathrm{H}},
\end{align}
where $(\widetilde{\mathbf{F}}_\mathrm{os}^\mathrm{H}\otimes \mathbf{I}_{N_\mathrm{t}})^{\dag}$ represents the Moore-Penrose pseudo-inverse of $\widetilde{\mathbf{F}}_\mathrm{os}^\mathrm{H}\otimes \mathbf{I}_{N_\mathrm{t}}$.
Hence, we can formulate the optimization problem similar as problem (P2) in \eqref{eq:OFDM04} by optimizing the DFRC MIMO-OFDM waveform matrix after IDFT operation, i.e., $\mathbf{G}=\mathbf{D}\mathbf{X}_\mathrm{s}$, with the consideration of the accurate PAPR constraints in  \eqref{PAPR_acc2}, which is expressed as
\begin{subeqnarray}\label{eq:OFDM048}
\hspace{-8mm}({\rm \widehat{P}1}): \underset{\mathbf{G}}{\min} &&\hspace{-6mm} \frac{\rho}{\left\|\mathbf{S}\right\|_\mathrm{F}^2}
\left\| \mathbf{H}_\mathrm{D}\mathbf{G}-\mathbf{S} \right\|_\mathrm{F}^2
+\frac{1-\rho}{\left\|\mathbf{G}_{0} \right\|_\mathrm{F}^2} \left\| \mathbf{G}-\mathbf{G}_0 \right\|_\mathrm{F}^2 \slabel{eq:OFDM048_0} \\
\mathrm{s.t.}
&&\hspace{-6mm} \big|[\mathbf{\Theta}\mathbf{G}]_{\gamma,l}\big|^2 \leq \frac{\varepsilon P_\mathrm{t}}{NN_\mathrm{\mathrm{t}}}, ~\forall \gamma \in N_\mathrm{\gamma st}, ~\forall l\in\mathcal{L},  \slabel{eq:OFDM048_1}\\
&&\hspace{-6mm} \parallel\hspace{-1mm} \mathbf{G} \hspace{-1mm}\parallel\hspace{-1mm}{_\mathrm{F}^2}=LP_\mathrm{t}^\mathrm{s}, \slabel{eq:OFDM048_2} \\
&&\hspace{-6mm} \left\|\mathbf{\Gamma}_\mathrm{c} \mathbf{G} \right\|_\mathrm{F}^2=LP_\mathrm{t}^\mathrm{c}, \slabel{eq:OFDM048_3}
\end{subeqnarray}
where $\mathbf{\Theta}=\mathbf{D}_\mathrm{os}\mathbf{D}^{-1} \in \mathbb{C}^{\Upsilon N_\mathrm{s}N_\mathrm{t}\times N_\mathrm{s}N_\mathrm{t} }$ with $\mathbf{D}_\mathrm{os}=\widetilde{\mathbf{F}}_\mathrm{os}^\mathrm{H}\otimes \mathbf{I}_{N_\mathrm{t}}$.
%
%
The PAPR constraints in \eqref{eq:OFDM048_1} can be equivalently transformed into the vector form with the optimization vector $\mathbf{g}$ as
\begin{align}\label{PAPR_Theta}
\mathrm{diag}(\widetilde{\mathbf{\Theta}}\mathbf{gg}^\mathrm{H}\widetilde{\mathbf{\Theta}}^\mathrm{H}) \leq \frac{\varepsilon P_\mathrm{t}}{NN_\mathrm{\mathrm{t}}}\mathbf{1}_{\Upsilon N_\mathrm{s}N_\mathrm{t}L}.
\end{align}
where $\widetilde{\mathbf{\Theta}}=\mathbf{I}_L\otimes\mathbf{\Theta} \in \mathbb{C}^{\Upsilon N_\mathrm{s}N_\mathrm{t}L \times N_\mathrm{s}N_\mathrm{t}L }$.

Similar to the case with Nyquist-rate sampling in Section \ref{sec:Problem_CP}, we can finally re-express problem ($\widehat{\mathrm{P}}$1) in \eqref{eq:OFDM048} as
\begin{subeqnarray}\label{eq:OFDM54}
({\rm \widehat{P}2}): \underset{\widehat{\mathbf{G}}}{\max} &&\hspace{-4mm} \mathrm{Tr}(\mathbf{Q}\widehat{\mathbf{G}})  \slabel{eq:OFDM54_0} \\
\mathrm{s.t.} ~
&&\hspace{-4mm} \mathrm{diag}(\widehat{\mathbf{\Theta}}\widehat{\mathbf{G}}\widehat{\mathbf{\Theta}}^\mathrm{H}) \leq \frac{\varepsilon P_\mathrm{t}}{NN_\mathrm{\mathrm{t}}}\mathbf{1}_{\Upsilon N_\mathrm{s}N_\mathrm{t}L},   \slabel{eq:OFDM54_1}\\ 
&&\hspace{-4mm} \mathrm{Tr}( \widehat{\mathbf{G}} )=LP_\mathrm{t}^\mathrm{s}+1, \slabel{eq:OFDM54_3} \\
&&\hspace{-4mm} \mathrm{Tr}( \widehat{\mathbf{\Gamma}}_\mathrm{c}\widehat{\mathbf{G}} )=LP_\mathrm{t}^\mathrm{c}, \slabel{eq:OFDM54_6} \\
&&\hspace{-4mm} \widehat{\mathbf{G}}_{N_\mathrm{s}N_\mathrm{t}L+1,N_\mathrm{s}N_\mathrm{t}L+1}=1,   \slabel{eq:OFDM54_2} \\
&&\hspace{-4mm} \widehat{\mathbf{G}}\succeq \mathbf{0},  \slabel{eq:OFDM54_7} \\ 
&&\hspace{-4mm} \mathrm{rank}(\widehat{\mathbf{G}})=1,\slabel{eq:OFDM54_4}
\end{subeqnarray}
where $\widehat{\mathbf{\Theta}}= \big[\widetilde{\mathbf{\Theta}} ~  \mathbf{0}_{\Upsilon N_\mathrm{s}N_\mathrm{t}L \times 1}\big] \in \mathbb{C}^{\Upsilon N_\mathrm{s}N_\mathrm{t}L \times (N_\mathrm{s}N_\mathrm{t}L+1)}$.
Problem ($\widehat{\mathrm{P}}$2) has a similar structure to problem (P5) in \eqref{eq:OFDM005} and can be optimally solved in a similar way.
Moreover, we can also divide the problem ($\widehat{\mathrm{P}}2$) in \eqref{eq:OFDM54} into $L$ subproblems each corresponding to one frame/snapshot as in Section \ref{sec:Waveform_Design_low}, and then leverage the low-complexity algorithm to obtain the low-PAPR DFRC MIMO-OFDM waveform solution by solving $L$ subproblems \mbox{parallelly}.

\section{Simulation Results}\label{sec:simulation}
In this section, simulation results are given to demonstrate the effectiveness of the proposed methods for designing the low-PAPR DFRC MIMO-OFDM waveform in both scenarios with Nyquist-rate sampling (NS) and oversampling (OS).
The performance comparison for these two cases are given to verify the necessity of utilizing oversampling in practical OFDM systems for measuring the PAPR levels of OFDM signals.
The performance results for communications and radar are analyzed with either communication priority or radar priority based on the weighting factor $\rho$, demonstrating the feasibility, effectiveness,  and flexibility of the proposed DFRC \mbox{waveform} design methods. Also, the  performance tradeoff between communications and radar are investigated to show the capability of the proposed waveform design methods in achieving satisfactory balance between these two functionalities.

In the following figures, the proposed low-PAPR DFRC MIMO-OFDM waveform design scheme is named as `DFRC' where we consider four different PAPR threshold values: 0.3dB, 1dB, 2dB, 3dB, and one benchmark without PAPR constraints; The benchmark scheme of desired radar waveform with directional beampattern design in noted as `Directional-Strict'; Another banchmark scheme of power constrained zero forcing (ZF) beamforming is denoted as `ZF-Power Constrained', where the power of transmitted ZF waveform is limited by the same power budget of the `DFRC' scheme.
The unit-power QPSK alphabet is utilized as the constellation for the communication UEs, and the SNR is defined as $P_\mathrm{t}/\sigma^2$. Other basic simulation parameters are listed in Table~\ref{table1} unless specified otherwise. The obtained results in the following figures are averaged over 1000 Monte Carlo simulations.

\setlength{\tabcolsep}{0.3 pt}\begin{table}[thb]
\centering
\caption{Simulation Parameters}\label{table1}
\vspace{-2mm}
{\footnotesize{
\begin{tabular}{|l|l|l|}
\hline
~\textbf{Parameter }&~{\textbf{Symbol}} &~{\textbf{Value}} \\
\hline
~Number of transmit antennas at the DFRC-BS~   &~$N_\mathrm{t}$  \quad\quad\quad&~8 \quad\quad\quad\\
\hline
~Number of subcarriers~   &~$N_\mathrm{s}$  &~16 \\
\hline
~Number of non-zero channel taps~   &~$U$  &~4 \\
\hline
~Length of CP~   &~$N_\mathrm{c}$  &~3 \\
\hline
~Number of active UEs~   &~$K$  &~2\\
\hline
~Number of targets~   &~$M$  &~3\\
\hline
~The frame length~   &~$L$  &~128\\
\hline
~Power budget for each OFDM symbol &~$P_\mathrm{t}$ &~1W  \\
\hline
~Oversampling rate &~$\Upsilon$ &~4 \\
\hline
\end{tabular}
}}
\end{table}

\vspace{-4mm}
\subsection{Performance Comparison for Scenarios with NS and OS}\label{sec:Perf_Comparisons}
In this subsection, we show the performance comparisons between the cases with NS and OS discussed in Section \ref{sec:Problem_CP} and Section \ref{sec:Problem_CPOS}, respectively.
In Fig. \ref{SumRate_rho_GS_GSOS}, Fig. \ref{SER_rho_GS_GSOS}, and Fig. \ref{PD_rho_GS_GSOS}, we depict the curves of the average sum rate, SER, and radar detection probability $\mathcal{P}_\mathrm{D}$ of all the schemes versus the weighting ratio parameter $\rho$. It is known that the radar function has higher priority as $\rho$ approaches to 0 while communication function dominates as $\rho$ approaches to 1.
\begin{figure}[htbp]
\centering
\includegraphics[scale=0.48]{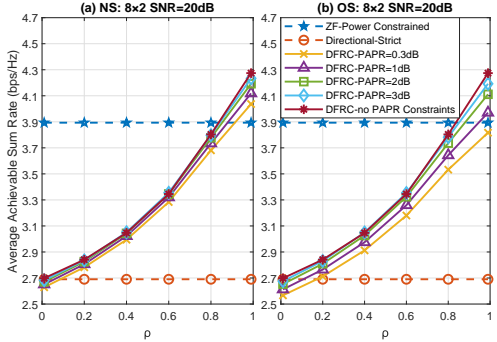}
\vspace{-2mm}
\caption{The average achievable sum rate versus $\rho$ for SNR=20dB.}
\label{SumRate_rho_GS_GSOS}
\end{figure}
\begin{figure}[htbp]
\centering
\includegraphics[scale=0.48]{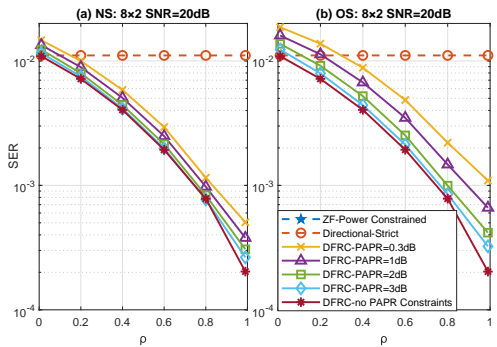}
\vspace{-2mm}
\caption{The average SER versus $\rho$ for SNR=20dB.}
\label{SER_rho_GS_GSOS}
\end{figure}

Fig. \ref{SumRate_rho_GS_GSOS} and Fig. \ref{SER_rho_GS_GSOS} present the variation of two important communication performance metrics: sum rate and SER versus $\rho$. From these two figures, we can observe that the average sum rates of the DFRC scheme increase with $\rho$ while the average SERs decrease with $\rho$ in both the NS and OS scenarios, which coincides the intuition that the communication priority increases with $\rho$ and thus better communication performance can be achieved with larger $\rho$.  The benchmark performance of the ZF-Power Constrained and the Directional-Strict schemes is not affected by $\rho$. It is easy to note that the communication performance of sum rates and SERs for the proposed DFRC scheme degrades to the Directional-Strict scheme when $\rho$ approaches 0, which operates a desired radar-only \mbox{waveform} design. Correspondingly, the communication performance approaches to the ZF-Power Constrained scheme as $\rho$ is close to 1,  operating a communication-only waveform design.\footnote{In Fig.  \ref{SER_rho_GS_GSOS}, the SER values of the ZF-Power Constrained schemes for the two cases are both 0, so the curves are not drawn with the log-scale y-axis.}
\begin{figure}[htbp]
\centering
\includegraphics[scale=0.48]{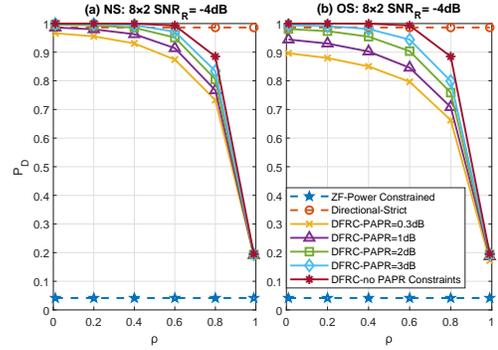}
\vspace{-2mm}
\caption{\hspace{-1mm}The average radar detection probability $\mathcal{P}_\mathrm{D}$ versus $\rho$ for SNR$_\mathrm{R}$=$-$4dB.}
\label{PD_rho_GS_GSOS}
\end{figure}

Fig. \ref{PD_rho_GS_GSOS} shows that the radar  detection probability values $\mathcal{P}_\mathrm{D}$ of the DFRC scheme decreases as  $\rho$ increases along with the design focus transforming from radar priority to  communication priority. Similarly, it demonstrates that the radar detection performance can catch up the Directional-Strict scheme when $\rho$ is close 0 while gradually degrades to the ZF-Power Constrained scheme as $\rho$ approaches 1. Note that the proposed DFRC scheme can achieve better radar detection performance (also better sum rate performance as shown in Fig. \ref{SumRate_rho_GS_GSOS}) than the ZF-Power Constrained scheme when $\rho$ approaches 1, demonstrating the effectiveness of the proposed  methods for low-PAPR DFRC MIMO-OFDM waveform design.

From Fig. \ref{SumRate_rho_GS_GSOS} to Fig. \ref{PD_rho_GS_GSOS}, we can see that better DFRC performance can be achieved with more relaxed PAPR constraints as PAPR threshold from 0.3dB to 3dB and even without PAPR constraints. Moreover, it is clearly shown that the gaps between the DFRC curves with  different PAPR constraints are much wider in the case with OS, indicating that the performance degradations caused by tightening the PAPR constraints are more obvious in the case with OS compared to NS. It is known that high PAPR is a vital disadvantage of the OFDM technique, and thus the PAPR levels of DFRC MIMO-OFDM waveforms should be carefully measured and then restrained  through effective low-PAPR DFRC MIMO-OFDM waveform design so as to meet the hardware requirements. In the case with NS, the measurement of PAPR is less accurate compared with OS due to insufficient sampling, resulting in a set of compact DFRC curves even with different PAPR constraints. These three figures verify the necessity of oversampling in measuring the accurate PAPR levels of OFDM waveforms in practical systems with low-cost non-linear power amplifiers.  Hence, in the following sub-sections, we will present more simulation results of low-PAPR MIMO-OFDM waveform design for the case with OS.

\vspace{-2mm}
\subsection{More Simulation Results for Communications and Radar  with Oversampling}\label{sec:Perf_Comms}
\subsubsection{Performance for Communications}
In Fig. \ref{SumRate_SER_SNR_rho08_GSOS} and Fig. \ref{SumRate_SER_SNR_rho02_GSOS}, the communication performance results of average sum rate and SER versus SNR are shown in scenarios with communication priority ($\rho=0.8$) and radar priority ($\rho=0.2$).
\begin{figure}[htbp]
\centering
\includegraphics[scale=0.48]{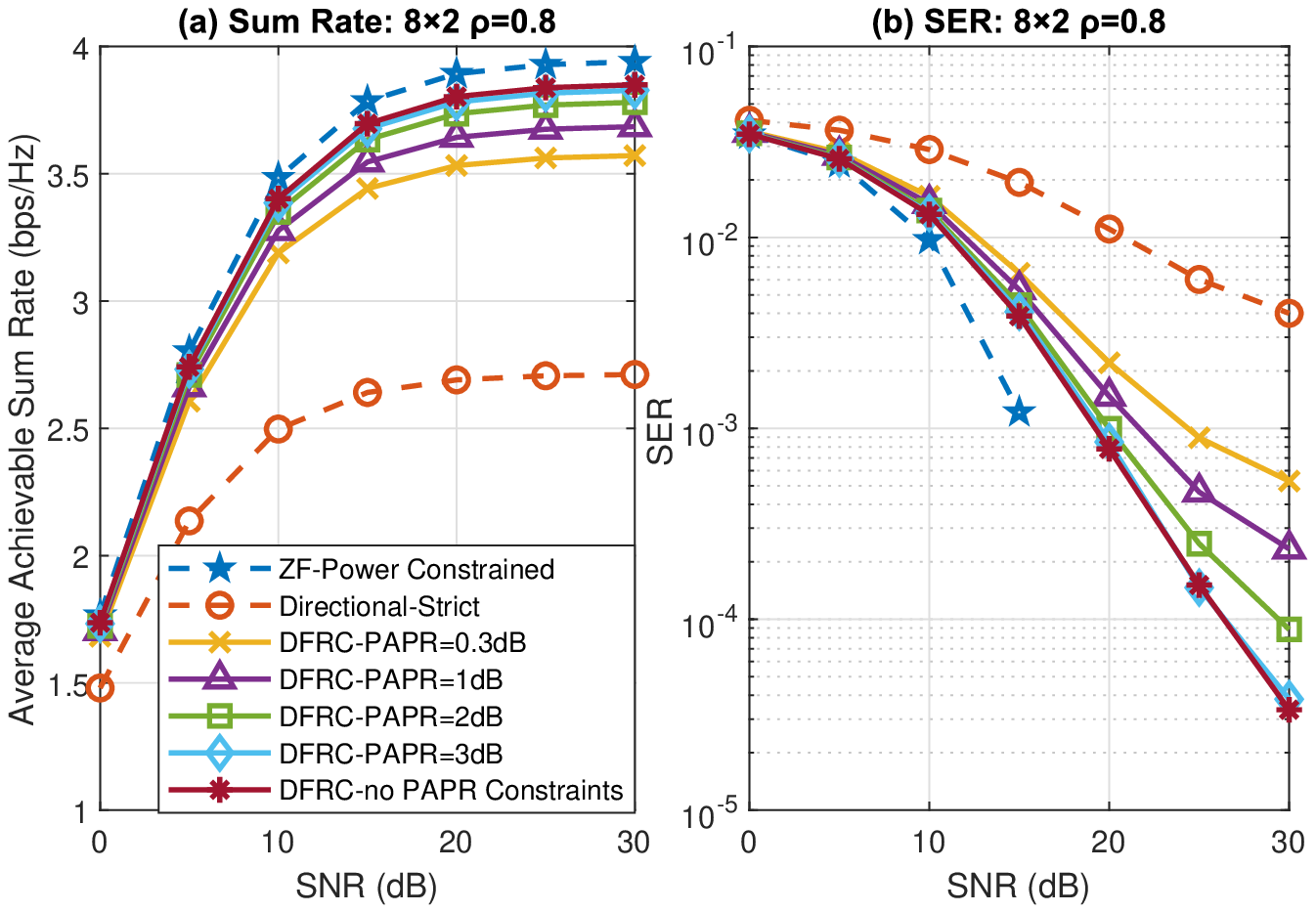}
\vspace{-2mm}
\caption{The average achievable sum rate and SER versus SNR for comunication priority with $\rho$=0.8.}
\label{SumRate_SER_SNR_rho08_GSOS}
\end{figure}
\begin{figure}[htbp]
\centering
\includegraphics[scale=0.48]{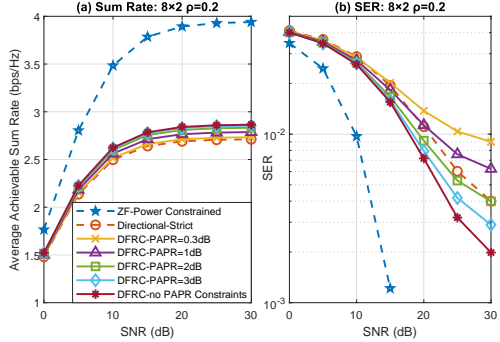}
\vspace{-2mm}
\caption{The average achievable sum rate and SER versus SNR for radar priority with $\rho$=0.2.}
\label{SumRate_SER_SNR_rho02_GSOS}
\end{figure}

From Fig. \ref{SumRate_SER_SNR_rho08_GSOS} and Fig. \ref{SumRate_SER_SNR_rho02_GSOS}, it is easy to observe that the communication-only ZF-Power Constrained scheme provides a upper bound while the radar-only Directional-Strict scheme provides a lower bound for communication performance in both scenarios with communication priority  ($\rho=0.8$)  and radar priority  ($\rho=0.2$). For the case with communication priority ($\rho=0.8$), the communication performance of sum rate and SER for the DFRC scheme is highly superior to the Directional-Strict scheme and close to the ZF-Power Constrained scheme. For the case with radar priority ($\rho=0.2$), the communication performance of sum rate degrades to the Directional-Strict scheme. The SER performance of the DFRC scheme with more restrict PAPR constraints such as 0.3dB and 1dB is worse than the Directional-Strict scheme since no PAPR constraints are executed for the Directional-Strict scheme. Both two figures show that average sum rates of all the schemes gradually saturate as the SNR increases due to the fact the noise power gradually becomes negligible as SNR increases considering the fixed limited signal power and thus MUI dominates leading to the sum rate saturation.\footnote{In order to make a fair comparison, we also limit the ZF signal power  with the same power budget $P_\mathrm{t}$ for each OFDM symbol, which restricts the performance of ZF scheme and leads to saturation of sum rate performance for ZF-Power Constrained scheme at high SNR.}
The values of SER decreases as SNR increases and the restriction effects of PAPR constraints on SER performance become more obvious with the enlarging of SNR.

\subsubsection{Performance for Radar}
In Fig. \ref{Beam_rho02_GSOS}  and Fig. \ref{Beam_rho08_GSOS}, the radar performance results of detection beampattern are presented in the scenarios with radar priority ($\rho=0.2$) and communication priority ($\rho=0.8$), where three targets of interest with angles of $-\pi/3$, 0 and $\pi/3$ are considered. In Fig. \ref{PD_rho_m2m4_GSOS}, we further show the radar detection probability $\mathcal{P}_\mathrm{D}$ versus $\rho$ in two scenarios with different radar received target echo SNR$_\mathrm{R}$.
\begin{figure}[htbp]
\centering
\includegraphics[scale=0.48]{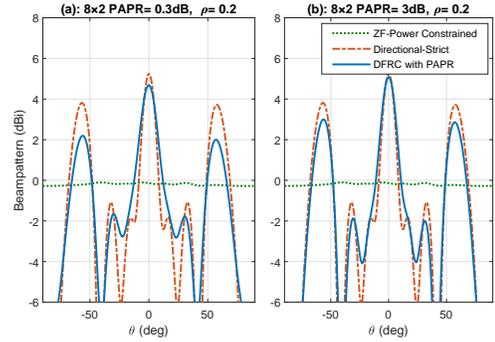}
\vspace{-2mm}
\caption{Radar beampatterns for radar priority with $\rho$=0.2.}
\label{Beam_rho02_GSOS}
\end{figure}
\begin{figure}[htbp]
\centering
\includegraphics[scale=0.48]{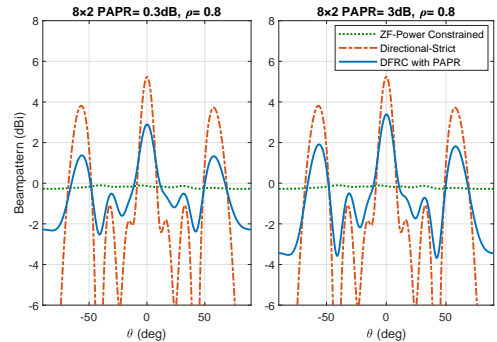}
\vspace{-2mm}
\caption{Radar beampatterns for communication priority with $\rho$=0.8.}
\label{Beam_rho08_GSOS}
\end{figure}

The performance of radar beampattern shown in Fig. \ref{Beam_rho02_GSOS} and Fig. \ref{Beam_rho08_GSOS} demonstrates the effectiveness of the designed low-PAPR DFRC MIMO-OFDM waveform in detecting targets, where we can observe that three targets can be clearly distinguished in both the cases  with radar priority  ($\rho=0.2$)  and communication priority  ($\rho=0.8$). For radar priority  ($\rho=0.2$), the beampattern of the proposed DFRC scheme can well match the desired Directional-Strict scheme and can be enhanced by relaxed PAPR constraints. Looser beampattern fit is achieved in the case of communication priority ($\rho=0.8$) but can still achieve satisfactory detection performance.  For the communication-only ZF-Power Constrained scheme, the beam gains at all directions are almost the same which is impossible to detect the interested targets.
From the comparison of the DFRC scheme between two sub-figures with different PAPR thresholds in Fig. \ref{Beam_rho02_GSOS}-\ref{Beam_rho08_GSOS}, we can see that better beampattern match can be achieved with relaxed PAPR constraints.
\begin{figure}[htbp]
\centering
\includegraphics[scale=0.48]{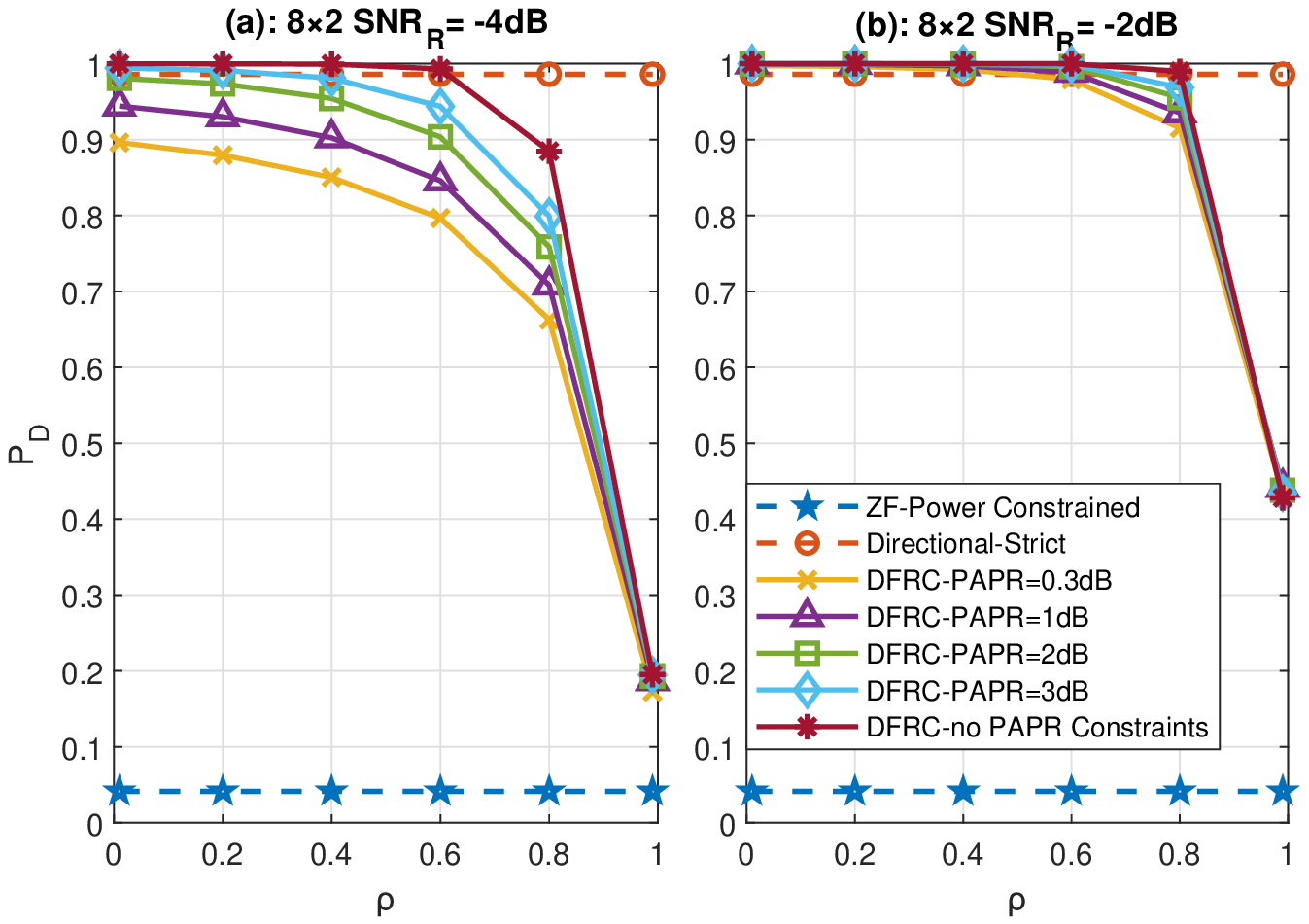}
\vspace{-2mm}
\caption{The average radar detection probability $\mathcal{P}_\mathrm{D}$ versus $\rho$ for SNR$_\mathrm{R}=-$4dB and SNR$_\mathrm{R}=-$2dB.}
\label{PD_rho_m2m4_GSOS}
\end{figure}

The performance results of radar detection probability $\mathcal{P}_\mathrm{D}$ versus $\rho$ are depicted in Fig. \ref{PD_rho_m2m4_GSOS} with SNR$_\mathrm{R}=-$4dB and SNR$_\mathrm{R}=-$2dB.
As a crucial radar performance metric, $\mathcal{P}_\mathrm{D}$ is highly affected by the  radar received target echo SNR$_\mathrm{R}$, which is described in Section \ref{Pd_SNRR}. The larger SNR$_\mathrm{R}$, the better radar detection performance can be achieved.
For the case of SNR$_\mathrm{R}=-$4dB, $\mathcal{P}_\mathrm{D}\geq0.8$ can be achieved by all the DFRC curves when $\rho\leq0.6$ even with the strictest PAPR constraint of 0.3dB, and $\mathcal{P}_\mathrm{D}\geq0.65$ for the case with communication priority $\rho=0.8$.
When SNR$_\mathrm{R}$ increases to $-$2dB, DFRC curves are capable to achieve $\mathcal{P}_\mathrm{D}\geq0.9$ for $\rho=0.8$ and the detection probability can always approach 1 as the Directional-Strict scheme  when $\rho\leq 0.6$.

\vspace{-2mm}
\subsection{Performance Tradeoff between Communications and Radar}\label{sec:Perf_Tradeoff}
In this section, we try to analyze the direct performance tradeoff between functionalities of communications and radar.
The communication results of average sum rate and SER versus the radar detection probability $\mathcal{P}_\mathrm{D}$ are depicted in Fig. \ref{SumRate_SER_PD_SNRRm4_GSOS} and Fig. \ref{SumRate_SER_PD_SNRRm2_GSOS} for SNR$_\mathrm{R}=-$4dB and SNR$_\mathrm{R}=-$2dB, respectively. The six points shown in these two figures correspond to $\rho=[0.01,0.2,0.4,0.6,0.8,0.99]$.
\begin{figure}[htbp]
\centering
\includegraphics[scale=0.48]{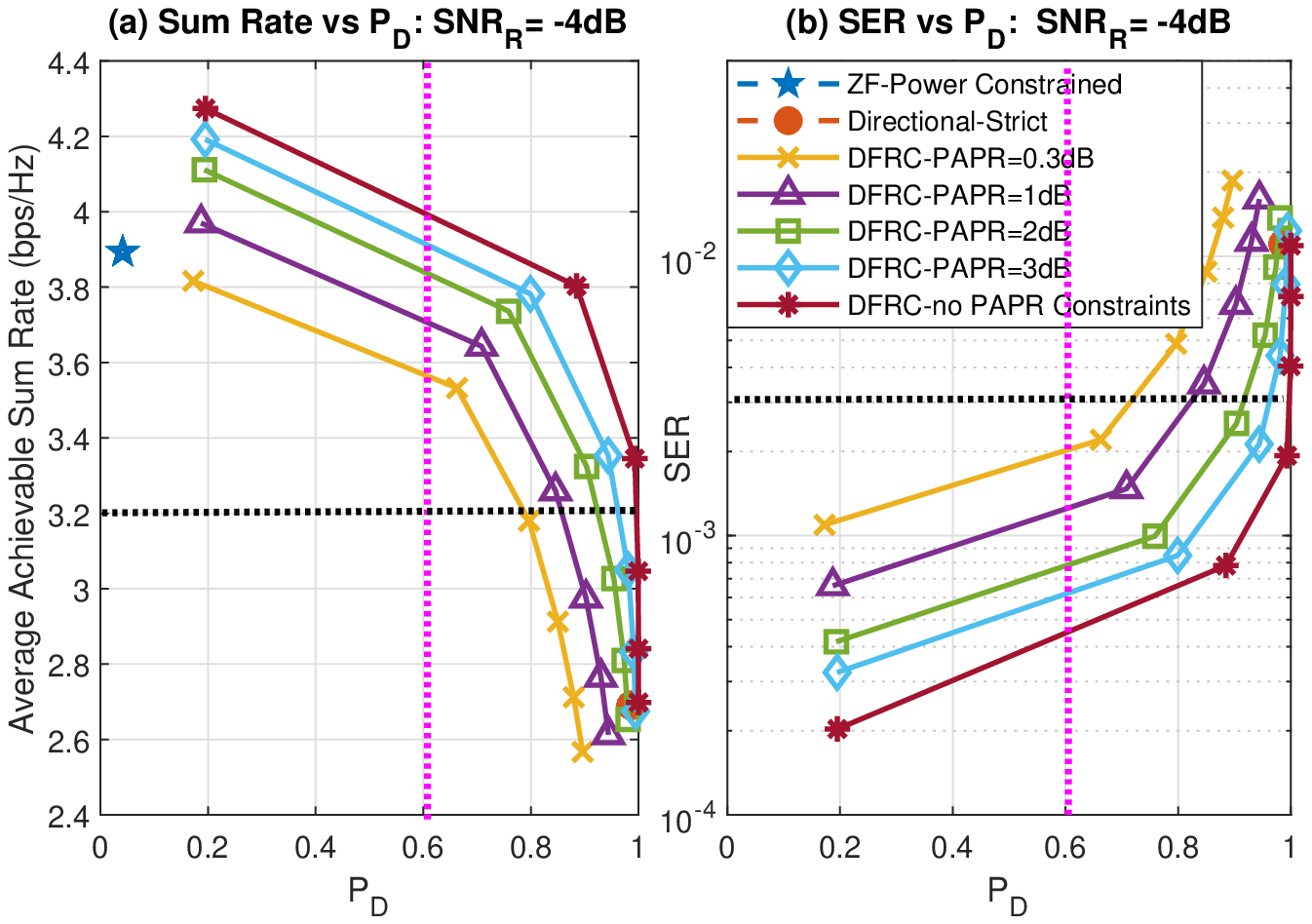}
\vspace{-2mm}
\caption{The average achievable communication sum rate and SER for SER=20dB versus the radar detection probability $\mathcal{P}_\mathrm{D}$ for SNR$_\mathrm{R}$=$-$4dB.}
\label{SumRate_SER_PD_SNRRm4_GSOS}
\end{figure}
\begin{figure}[htbp]
\centering
\includegraphics[scale=0.48]{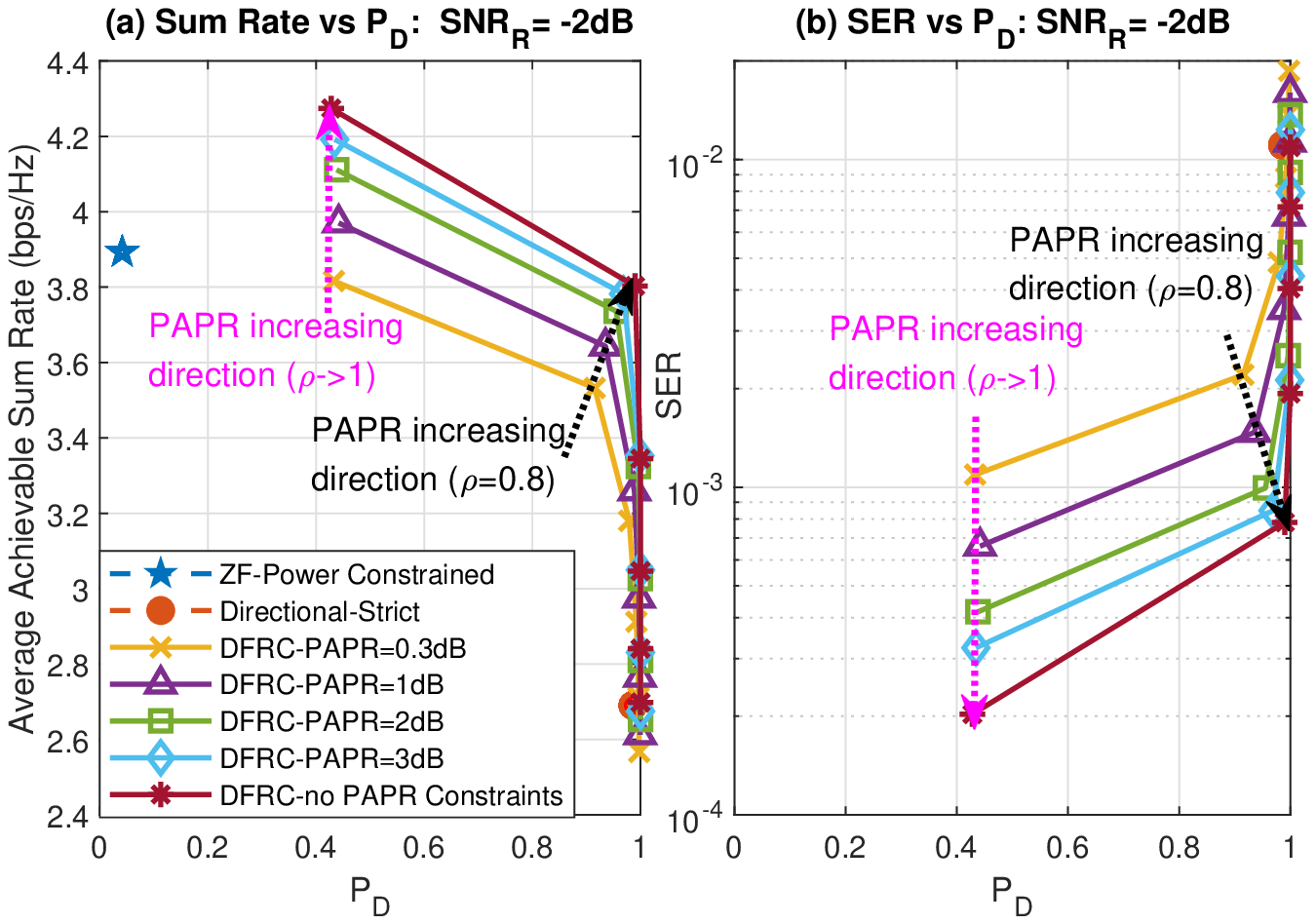}
\vspace{-2mm}
\caption{The average achievable communication sum rate and SER for SER=20dB versus the radar detection probability $\mathcal{P}_\mathrm{D}$ for SNR$_\mathrm{R}$=$-$2dB.}
\label{SumRate_SER_PD_SNRRm2_GSOS}
\end{figure}

In Fig. \ref{SumRate_SER_PD_SNRRm4_GSOS}, we draw two required thresholds for sum rate and $\mathcal{P}_\mathrm{D}$ in sub-figure (a) and two thresholds for SER and $\mathcal{P}_\mathrm{D}$ in sub-figure (b). For requirements of ($\mathcal{P}_\mathrm{D}$,sum rate), the points located in the top right corner of sub-figure (a) are satisfactory. In contrast, the points in the bottom right corner of sub-figure (b) are satisfactory for the ($\mathcal{P}_\mathrm{D}$,SER)  requirements. Giving the requirements of these three performance metrics for communications and radar, we can effectively set parameter $\rho$ satisfying certain PAPR constraints (hardware requirements).

In Fig. \ref{SumRate_SER_PD_SNRRm2_GSOS} with SNR$_\mathrm{R}=-$2dB, the $\mathcal{P}_\mathrm{D}$ requirements are easier to be satisfied since most of the points are located in the area of $\mathcal{P}_\mathrm{D}\geq0.8$, and thus the satisfactory settings of $\rho$ and PAPR thresholds can be flexibly selected based on the requirements of sum rate and SER. 
\begin{figure}[htbp]
\centering
\includegraphics[scale=0.48]{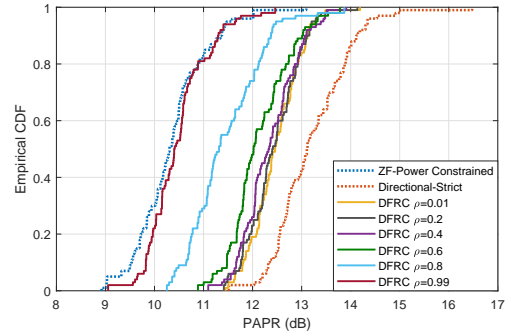}
\vspace{-2mm}
\caption{The empirical CDF of the PAPR (dB) without considering the PAPR constraints.}
\label{CDF_noPAPR}
\end{figure}

In Fig. \ref{CDF_noPAPR}, the empirical CDFs of waveform PAPR  for the proposed DFRC  scheme with different $\rho$ values  as well as the two benchmarks are provided, where the PAPR constraints are not considered. It is clear to see that the PAPR of the proposed DFRC scheme is within the range of [9,14]dB if no PAPR constraints are used. For radar priority of $\rho$ approaching 0, the PAPR levels with average value of 12.5dB get close to the radar-only Directional-Strict scheme with average value of 13.5dB. For communication priority of $\rho$ close to 1, the PAPR levels approach to the communication-only ZF-Power Constrained scheme with average value of 10.5dB. All these PAPR levels are too high to satisfy the strict hardware requirements in practical communication systems. In comparison with the simulation results in the previous figures, we can clearly observe that we can restrain the PAPR levels of the DFRC waveform to 3dB, 2dB, 1dB and even 0.3dB without too much performance degradation for both communications and radar, demonstrating the feasibility and effectiveness of the proposed low-PAPR DFRC MIMO-OFDM waveform design methods.

\vspace{-1mm}
\section{Conclusion}\label{sec:conclusion}
In this paper, we investigate the low-PAPR DFRC waveform design for MIMO-OFDM systems, where the scenario with Nyquist-rate sampling and oversampling are considered for measuring the waveform PAPR levels. A weighted objective function on normalized communication performance metric MUI and radar performance metric (distance to desired radar waveform) is minimized under the transmit power and PAPR constraints. The optimization problems can be transformed into standard SDP, then the SDR method is leveraged to find the optimal solution where the rank-1 constraint is satisfied in general. In addition, low-complexity methods are provided to reduce the overload for solving the SDP problems where the original problems are divided into subproblems corresponding to each frame/snapshot. The simulation results demonstrate that oversampling on OFDM signals can provide an accurate way for measuring the PAPR levels. Moreover, the feasibility, effectiveness, and flexibility of the proposed low-PAPR DFRC MIMO-OFDM waveform  design methods are verified by sufficient numerical simulations.

\vspace{-2mm}
\appendices
\section{Proof of Lemma~\ref{lemma0}}\label{sec:proof_lem0}
Considering the IDFT operation applied at the DFRC-BS transmitter  and the DFT processing operated at the downlink UE receivers, the received baseband signal of  UE $k$ during frame $l \in \mathcal{L}$  in the frequency domain can be  expressed as 
\begin{align}\label{Eff_ykl1}
\mathbf{y}_k^l=\mathbf{F}_\mathrm{s} \mathbf{H}_{k}^\mathrm{s}(\mathbf{F}_\mathrm{s}^\mathrm{H}\otimes \mathbf{I}_{N_\mathrm{t}})\mathbf{x}_\mathrm{s}^l \in \mathbb{C}^{N_\mathrm{s}\times 1}.
\end{align}
where $\mathbf{x}_\mathrm{s}^l \in \mathbb{C}^{N_\mathrm{s}N_\mathrm{t} \times 1}$ is the $l$-th column of $\mathbf{X}_\mathrm{s}$ and $\mathbf{H}_{k}^\mathrm{s} \in \mathbb{C}^{N_\mathrm{s}\times N_\mathrm{s}N_\mathrm{t}}$ is a \emph{block circulant channel matrix} given as
{\small{
\begin{align} \label{H_cir_Compact}
 \mathbf{H}_{k}^\mathrm{s}=\left[
 \begin{matrix}
 \widetilde{\mathbf{h}}_{k,0}   &\hspace{-2mm}\mathbf{0}                     &\hspace{-2mm}\cdots  &\hspace{-2mm}\mathbf{0}  &\hspace{-2mm}\widetilde{\mathbf{h}}_{k,U-1}                             &\hspace{-2mm}\cdots    &\hspace{-2mm}\widetilde{\mathbf{h}}_{k,1}\\
 \widetilde{\mathbf{h}}_{k,1}   &\hspace{-2mm}\widetilde{\mathbf{h}}_{k,0}   &\hspace{-2mm}\cdots  &\hspace{-2mm}\mathbf{0}  &\hspace{-2mm}\mathbf{0}      &\hspace{-2mm}\cdots    &\hspace{-2mm}\widetilde{\mathbf{h}}_{k,2}\\
 \vdots                         &\hspace{-2mm}\vdots                         &\hspace{-2mm}\ddots  &\hspace{-2mm}\vdots                         &\hspace{-2mm}\vdots          &\hspace{-2mm}\ddots    &\hspace{-2mm}\vdots \\
 \widetilde{\mathbf{h}}_{k,U-1} &\hspace{-2mm}\widetilde{\mathbf{h}}_{k,U-2} &\hspace{-2mm}\cdots  &\hspace{-2mm}\widetilde{\mathbf{h}}_{k,0}   &\hspace{-2mm}\mathbf{0}                &\hspace{-2mm}\cdots    &\hspace{-2mm}\mathbf{0} \\
 \mathbf{0}                     &\hspace{-2mm}\widetilde{\mathbf{h}}_{k,U-1} &\hspace{-2mm}\cdots  &\hspace{-2mm}\widetilde{\mathbf{h}}_{k,1}   &\hspace{-2mm}\widetilde{\mathbf{h}}_{k,0}    &\hspace{-2mm}\cdots    &\hspace{-2mm}\mathbf{0} \\
 \vdots                         &\hspace{-2mm}\vdots                         &\hspace{-2mm}\ddots  &\hspace{-2mm}\vdots                         &\hspace{-2mm}\vdots
  &\hspace{-2mm}\ddots   &\hspace{-2mm}\mathbf{0} \\
 \mathbf{0}                     &\hspace{-2mm}\mathbf{0}                     &\hspace{-2mm}\cdots  &\hspace{-2mm}\widetilde{\mathbf{h}}_{k,U-1} &\hspace{-2mm}\widetilde{\mathbf{h}}_{k,U-2}
 &\hspace{-2mm}\cdots    &\hspace{-2mm}\widetilde{\mathbf{h}}_{k,0}
  \end{matrix}
  \right],
\end{align}
}}
thanks to the operation of adding CP.
We further denote the overall effective downlink channel matrix of UE $k$  as 
\begin{align}
\mathbf{h}_k&\triangleq\mathbf{F}_\mathrm{s} \mathbf{H}_{k}^\mathrm{s}(\mathbf{F}_\mathrm{s}^\mathrm{H}\otimes \mathbf{I}_{N_\mathrm{t}})\in \mathbb{C}^{N_\mathrm{s}\times N_\mathrm{s}N_\mathrm{t}} \nonumber \\
&=\mathbf{F}_\mathrm{s} \mathbf{H}_{k}^\mathrm{s}\mathbf{\Omega}\mathbf{\Omega}^\mathrm{T}(\mathbf{F}_\mathrm{s}^\mathrm{H}\otimes \mathbf{I}_{N_\mathrm{t}})\mathbf{\Omega}\mathbf{\Omega}^\mathrm{T} \nonumber \\
&\overset{(a)}{=}\mathbf{F}_\mathrm{s}[ \mathbf{H}_{k,1}^\mathrm{s},\cdots, \mathbf{H}_{k,N_\mathrm{t}}^\mathrm{s}]\mathbf{\Omega}^\mathrm{T}(\mathbf{F}_\mathrm{s}^\mathrm{H}\otimes \mathbf{I}_{N_\mathrm{t}})\mathbf{\Omega}\mathbf{\Omega}^\mathrm{T} \nonumber \\
&\overset{(b)}{=}\mathbf{F}_\mathrm{s} [\mathbf{H}_{k,1}^\mathrm{s}, \cdots, \mathbf{H}_{k,N_\mathrm{t}}^\mathrm{s}](\mathbf{I}_{N_\mathrm{t}} \otimes \mathbf{F}_\mathrm{s}^\mathrm{H}) \mathbf{\Omega}^\mathrm{T} \nonumber \\
&\overset{(c)}{=}[\mathbf{F}_\mathrm{s} \mathbf{H}_{k,1}^\mathrm{s} \mathbf{F}_\mathrm{s}^\mathrm{H}, \cdots, \mathbf{F}_\mathrm{s} \mathbf{H}_{k,N_\mathrm{t}}^\mathrm{s} \mathbf{F}_\mathrm{s}^\mathrm{H}]\mathbf{\Omega}^\mathrm{T} \nonumber \\
&\overset{(d)}{=}[\mathbf{\Lambda}_{k,1}, \cdots,\mathbf{\Lambda}_{k,N_\mathrm{t}}]\mathbf{\Omega}^\mathrm{T} \nonumber \\
&\overset{(e)}{=}\mathrm{diag}(\mathbf{h}_{k,1}, \cdots, \mathbf{h}_{k,N_\mathrm{s}}) \in \mathbb{C}^{N_\mathrm{s}\times N_\mathrm{s}N_\mathrm{t}},
\end{align}
where (a) holds by column permutation $\mathbf{H}_{k}^\mathrm{s}\mathbf{\Omega}$ with $\mathbf{\Omega}\mathbf{\Omega}^\mathrm{T}=\mathbf{I}_{N_\mathrm{s}N_\mathrm{t}}$ that  $\mathbf{H}_{k}^\mathrm{s}$ is transformed into $N_\mathrm{t}$ circulant channel matrix $\mathbf{H}_{k,n_\mathrm{t}}^\mathrm{s} \in \mathbb{C}^{N_\mathrm{s}\times N_\mathrm{s}}$ for $n_\mathrm{t}\in \mathcal{N}_\mathrm{t}=\{1,\cdots,N_\mathrm{t}\}$.
In addition, through both raw and column permutations, we have  $\mathbf{\Omega}^\mathrm{T}(\mathbf{F}_\mathrm{s}^\mathrm{H}\otimes \mathbf{I}_{N_\mathrm{t}})\mathbf{\Omega}=(\mathbf{I}_{N_\mathrm{t}} \otimes \mathbf{F}_\mathrm{s}^\mathrm{H})$, leading to (b).
The equation of (c) is based on the special structure of $(\mathbf{I}_{N_\mathrm{t}} \otimes \mathbf{F}_\mathrm{s}^\mathrm{H})$.
Then $\mathbf{F}_\mathrm{s} \mathbf{H}_{k,n_\mathrm{t}}^\mathrm{s} \mathbf{F}_\mathrm{s}^\mathrm{H}=\mathbf{\Lambda}_{k,n_\mathrm{t}}$ as (d) \cite{TCOM2002_B.Muquet_Cyclic}, where $\mathbf{\Lambda}_{k,n_\mathrm{t}} \in \mathbb{C}^{N_\mathrm{s}\times N_\mathrm{s}}$ is a diagonal matrix with elements $\Lambda_{k,n_\mathrm{t},n}=\sum_{u=0}^{U-1}\widetilde{h}_{k,u,n_\mathrm{t}}e^{-\frac{j2\pi u(n-1)}{N_\mathrm{s}}}$ for  $n_\mathrm{t} \in \mathcal{N}_\mathrm{t}$ and $n\in \mathcal{N}_\mathrm{s}$. After the final column permutation on $[\mathbf{\Lambda}_{k,1}, \cdots,\mathbf{\Lambda}_{k,N_\mathrm{t}}]$ with $\mathbf{\Omega}^\mathrm{T}$ in (d), a diagonal structure can be obtained as in (e) where 
\begin{align}
&\mathbf{h}_{k,n}=[\Lambda_{k,1,n}, \cdots, \Lambda_{k,N_\mathrm{t},n}] \nonumber \\
=&\sum_{u=0}^{U-1}[\widetilde{h}_{k,u,1}, \cdots, \widetilde{h}_{k,u,N_\mathrm{t}}]e^{-\frac{j2\pi u(n-1)}{N_\mathrm{s}}} \nonumber \\
=&\sum_{u=0}^{U-1}\widetilde{\mathbf{h}}_{k,u}e^{-\frac{j2\pi u(n-1)}{N_\mathrm{s}}} \in \mathbb{C}^{1\times N_\mathrm{t}},~ n \in \mathcal{N}_\mathrm{s}, ~k \in \mathcal{K}.
\end{align}
Hence, combining \eqref{Eff_ykl1}, we have
\begin{align}
\mathbf{y}_k^l=\mathbf{h}_k\mathbf{x}_\mathrm{s}^l=\mathrm{diag}(\mathbf{h}_{k,1}, \cdots, \mathbf{h}_{k,N_\mathrm{s}})\mathbf{x}_\mathrm{s}^l.
\end{align}
Further considering all the $K$ users and $L$ symbol frames, a compact form of the noiseless received signal can be given as
\begin{align}
\mathbf{Y}_\mathrm{s}=\mathbf{H}_\mathrm{s}\mathbf{X}_\mathrm{s}\in \mathbb{C}^{N_\mathrm{s}K\times L},
\end{align}
where $\mathbf{H}_\mathrm{s}=\mathrm{diag}(\mathbf{H}_1,\cdots, \mathbf{H}_{N_\mathrm{s}}) \in \mathbb{C}^{N_\mathrm{s}K\times N_\mathrm{s}N_\mathrm{t}}$
with $\mathbf{H}_n=[{\mathbf{h}}_{1,n}^\mathrm{T}, \cdots, {\mathbf{h}}_{K,n}^\mathrm{T}]^\mathrm{T} \in \mathbb{C}^{K\times N_\mathrm{t}}$.
The proof of Lemma \ref{lemma0} has been completed.

\vspace{-2mm}
\section{Proof of Lemma~\ref{lemma1}}\label{sec:proof_lem1}

It is easy to observe that the problem (P5) in \eqref{eq:OFDM005} ignoring the rank-1 constraint \eqref{eq:OFDM005_4} is a convex optimization problem, and thus the optimal solution must satisfy the Karush-Kuhn-Tacker (KKT) conditions \cite{B_Boyd04Convex}. We first express the Lagrangian function of such problem as
\vspace{-2mm}

{\small{
\begin{align}
\mathcal{L}=&\mathrm{Tr}\big(\mathbf{Q}\widehat{\mathbf{G}}\big)+\sum_{q=1}^{N_\mathrm{s}N_\mathrm{t}L}
\lambda_q\left( \mathrm{Tr}\big(\mathbf{\Pi}_q\widehat{\mathbf{G}} \big)-\frac{\varepsilon P_\mathrm{t}}{NN_\mathrm{\mathrm{t}}} \right) \nonumber \\
&+\nu\left(\mathrm{Tr}\big( \widehat{\mathbf{G}} \big)-LP_\mathrm{t}^\mathrm{s}-1 \right)
+\varrho\left(\mathrm{Tr} \big( \widehat{\mathbf{\Gamma}}_\mathrm{c}\widehat{\mathbf{G}} \big)-LP_\mathrm{t}^\mathrm{c} \right) \nonumber \\
&-\mathrm{Tr}\big(\mathbf{\Psi}\widehat{\mathbf{G}}\big)
-\varphi\left( \mathrm{Tr}\big(\mathbf{\Pi}_{N_\mathrm{s}N_\mathrm{t}L+1}\widehat{\mathbf{G}} \big)-1\right),
\end{align}
}}
\hspace{-2mm}where $\lambda_q\geq0$ for $q \in \mathcal{N}_\mathrm{stL}$ and $\mathbf{\Psi}\succeq0$ are the Lagrangian multipliers corresponding to the inequality constraints \eqref{eq:OFDM005_1} and \eqref{eq:OFDM005_7}, while $\nu$, $\varrho$ and $\varphi$ are the Lagrangian multipliers corresponding to the equality constraints \eqref{eq:OFDM005_3}, \eqref{eq:OFDM005_6}, and \eqref{eq:OFDM005_2}. For the global optimal solution of the considered problem, the optimal Lagrangian multipliers are uniquely determined.

Here, $\mathbf{\Pi}_q \in \mathbb{R}^{{N_\mathrm{s}N_\mathrm{t}L+1}}$ is a diagonal matrix where only the $q$-th  element on the diagonal line is non-zero with value of 1. Hence, we can obtain the KKT conditions given below
\vspace{-4mm}

{\small{
\begin{align}
&\hspace{-1mm}\frac{\partial \mathcal{L}}{\partial \widehat{\mathbf{G}}}=\mathbf{Q}+\hspace{-2mm}\sum_{q=1}^{N_\mathrm{s}N_\mathrm{t}L}\hspace{-2mm}\lambda_q \mathbf{\Pi}_q +\nu \mathbf{I} +\varrho\widehat{\mathbf{\Gamma}}_\mathrm{c}-\mathbf{\Psi}-\varphi\mathbf{\Pi}_{N_\mathrm{s}N_\mathrm{t}L+1}=0, \label{Lag_G} \\
&\hspace{-1mm}\lambda_q\left( \mathrm{Tr}\big(\mathbf{\Pi}_q\widehat{\mathbf{G}} \big)-\frac{\varepsilon P_\mathrm{t}}{NN_\mathrm{\mathrm{t}}} \right)=0, ~\lambda_q\geq0, ~q \in \mathcal{N}_\mathrm{stL}, \\
&\hspace{-1mm}\mathrm{Tr}\big(\mathbf{\Psi}\widehat{\mathbf{G}}\big)=0, ~\mathbf{\Psi}\succeq\mathbf{0}\\
&\hspace{-1mm}\mathrm{Tr}\big( \widehat{\mathbf{G}} \big)-LP_\mathrm{t}^\mathrm{s}-1=0, \\
&\hspace{-1mm}\mathrm{Tr} \big( \widehat{\mathbf{\Gamma}}_\mathrm{c}\widehat{\mathbf{G}} \big)-LP_\mathrm{t}^\mathrm{c} =0, \\
&\hspace{-1mm}\mathrm{Tr}\big(\mathbf{\Pi}_{N_\mathrm{s}N_\mathrm{t}L+1}\widehat{\mathbf{G}} \big)-1=0.
\end{align}
}}
\hspace{-1.5mm}The above KKT conditions should be satisfied with the global optimal solution of problem (P5) in \eqref{eq:OFDM005} ignoring the rank-1 constraint. We denote the optimal solution  as $\widehat{\mathbf{G}}^\mathrm{o}$, and the corresponding optimal Lagrangian multipliers are denoted as $\lambda_q^\mathrm{o}$, $\nu^\mathrm{o}$, $\varrho^\mathrm{o}$, $\varphi^\mathrm{o}$, and $\mathbf{\Psi}^\mathrm{o}$.
From the condition \eqref{Lag_G}, we have
\begin{align}\label{Omega_opt}
\mathbf{\Psi}^\mathrm{o}&=\mathbf{Q}+\hspace{-2mm}\sum_{q=1}^{N_\mathrm{s}N_\mathrm{t}L}\hspace{-2mm}\lambda_q^\mathrm{o} \mathbf{\Pi}_q +\nu^\mathrm{o} \mathbf{I} +\varrho^\mathrm{o}\widehat{\mathbf{\Gamma}}_\mathrm{c}-\varphi^\mathrm{o}\mathbf{\Pi}_{N_\mathrm{s}N_\mathrm{t}L+1} \nonumber \\
&=\mathbf{Q}+\mathbf{\Pi}^\mathrm{o}+\nu^\mathrm{o} \mathbf{I} +\varrho^\mathrm{o}\widehat{\mathbf{\Gamma}}_\mathrm{c}-\varphi^\mathrm{o}\mathbf{\Pi}_{N_\mathrm{s}N_\mathrm{t}L+1},
\end{align}
where $\mathbf{\Pi}^\mathrm{o}\succeq \mathbf{0}$ is a diagonal matrix with 
\begin{align}
\mathrm{diag}(\mathbf{\Pi}^\mathrm{o})=[\lambda_1^\mathrm{o},\lambda_2^\mathrm{o},\cdots,\lambda_{N_\mathrm{s}N_\mathrm{t}L}^\mathrm{o},0]^\mathrm{T}
\end{align}
and the rank of $\mathbf{\Pi}^\mathrm{o}$ is determined by the number of non-zero Lagrangian multipliers $\lambda_q^\mathrm{o}$ for  $q \in \mathcal{N}_\mathrm{stL}$.
It is easy to verify that $\widehat{\mathbf{\Gamma}}_\mathrm{c}\succeq \mathbf{0}$ is also a diagonal matrix with $N_\mathrm{c}N_\mathrm{t}L$ non-zero elements which with the same value of 1, and thus $\mathrm{rank}(\widehat{\mathbf{\Gamma}}_\mathrm{c})=N_\mathrm{c}N_\mathrm{t}L$.
Hence, we can obtain
\begin{align}\label{Pi_I}
\mathbf{\Pi}^\mathrm{o}+\nu^\mathrm{o} \mathbf{I} +\varrho^\mathrm{o}\widehat{\mathbf{\Gamma}}_\mathrm{c}\succ \mathbf{0},
\end{align}
by considering the fact  $\mathbf{\Pi}^\mathrm{o}$, $\nu^\mathrm{o}$ and $\varrho^\mathrm{o}$ are uniquely determined and the special case of $\nu^\mathrm{o}=0$ rarely occurs.
Based on the definition of $\mathbf{Q}$ in \eqref{Matrix_Q} and the non-negative property of objective function  \eqref{eq:OFDM043_0}, we know that $\mathbf{Q}\succeq 0$.
Hence, we can further derive that
\begin{align}
\mathbf{\Xi}\triangleq\mathbf{Q}+\mathbf{\Pi}^\mathrm{o}+\nu^\mathrm{o} \mathbf{I} +\varrho^\mathrm{o}\widehat{\mathbf{\Gamma}}_\mathrm{c}\succ \mathbf{0},
\end{align}
which can be verified by contradiction. Assuming that $\mathbf{\Xi}\preceq \mathbf{0}$, there exists at least one vector $\mathbf{x} \in \mathbb{C}^{(N_\mathrm{s}N_\mathrm{t}L+1)\times1}$ such that $\mathbf{x}^\mathrm{H} \mathbf{\Xi} \mathbf{x}\leq 0$, also we have
\begin{align}
\mathbf{x}^\mathrm{H} \mathbf{\Psi}^\mathrm{o} \mathbf{x}=\mathbf{x}^\mathrm{H} \mathbf{\Xi} \mathbf{x}-\varphi^\mathrm{o}  \mathbf{x}^\mathrm{H} \mathbf{\Pi}_{N_\mathrm{s}N_\mathrm{t}L+1} \mathbf{x}\geq 0
\end{align}
due to the fact that $\mathbf{\Psi}^\mathrm{o} \succeq \mathbf{0}$. Considering that $\mathbf{\Pi}_{N_\mathrm{s}N_\mathrm{t}L+1} \succeq \mathbf{0}$ as well, we have $\mathbf{x}^\mathrm{H} \mathbf{\Pi}_{N_\mathrm{s}N_\mathrm{t}L+1} \mathbf{x}\geq0$, leading to the results of $\mathbf{x}^\mathrm{H} \mathbf{\Psi}^\mathrm{o} \mathbf{x}=0$, $\mathbf{x}^\mathrm{H} \mathbf{\Pi}_{N_\mathrm{s}N_\mathrm{t}L+1} \mathbf{x}=0$ and $\mathbf{x}^\mathrm{H} \mathbf{\Xi} \mathbf{x}=0$. In addition, $\mathbf{x}^\mathrm{H} \mathbf{Q} \mathbf{x}\geq0$, and thus $\mathbf{x}^\mathrm{H}(\mathbf{\Pi}^\mathrm{o}+\nu^\mathrm{o} \mathbf{I} +\varrho^\mathrm{o}\widehat{\mathbf{\Gamma}}_\mathrm{c})\mathbf{x}=\mathbf{x}^\mathrm{H} \mathbf{\Xi} \mathbf{x}-\mathbf{x}^\mathrm{H} \mathbf{Q} \mathbf{x}\leq0$, which contradicts the result in \eqref{Pi_I}.

Based on the expression \eqref{Omega_opt}, we further have
\begin{align}
\mathrm{rank}(\mathbf{\Psi})&\geq \mathrm{rank}(\mathbf{\Xi})-\mathrm{rank}(\mathbf{\mathbf{\Pi}_{N_\mathrm{s}N_\mathrm{t}L+1}})\nonumber\\
&\geq N_\mathrm{s}N_\mathrm{t}L+1-1=N_\mathrm{s}N_\mathrm{t}L.
\end{align}
Since $\mathrm{Tr}\big(\mathbf{\Psi}\widehat{\mathbf{G}}\big)=0$ with $\mathbf{\Psi}\succeq\mathbf{0}$ and $\widehat{\mathbf{G}}\succeq\mathbf{0}$, it follows that $\mathbf{\Psi}\widehat{\mathbf{G}}=\mathbf{0}$ and $\mathrm{rank}(\mathbf{\Psi})+\mathrm{rank}(\widehat{\mathbf{G}})\leq N_\mathrm{s}N_\mathrm{t}L+1$, and thus
\begin{align}
\mathrm{rank}(\widehat{\mathbf{G}})\leq N_\mathrm{s}N_\mathrm{t}L+1-\mathrm{rank}(\mathbf{\Psi})\leq1.
^{}\end{align}
Also, $\mathrm{rank}(\widehat{\mathbf{G}})\geq 1$ as it is a non-zero matrix, leading to the final result of $\mathrm{rank}(\widehat{\mathbf{G}})=1$.

\bibliographystyle{IEEEtran}
\bibliography{DFRC_OFDM}

\end{document}